\documentclass[9]{article}%
\usepackage{amssymb}
\usepackage{amsfonts}
\usepackage{amsmath}
\usepackage{graphicx}%
\newtheorem{theorem}{Theorem}

\newtheorem{definition}[theorem]{Definition}

\newtheorem{lemma}[theorem]{Lemma}
\newtheorem{notation}[theorem]{Notation}

\newtheorem{proposition}[theorem]{Proposition}

\newenvironment{proof}[1][Proof]{\noindent\textbf{#1.} }{\ \rule{0.5em}{0.5em}}
\begin{document}

\title{The Underlying Dynamics of Credit Correlations}
\author{Arthur Berd\thanks{BlueMountain Capital Management. }
\and Robert Engle\thanks{Department of Finance, Stern School of Business, New York
University.}
\and Artem Voronov\thanks{Department of Economics, New York University.}}
\date{First draft: June, 2005\\
This draft: April, 2007}
\maketitle

\begin{abstract}
We propose a hybrid model of portfolio credit risk where the dynamics of the
underlying latent variables is governed by a one factor GARCH process. The
distinctive feature of such processes is that the long-term aggregate return
distributions can substantially deviate from the asymptotic Gaussian limit for
very long horizons. We introduce the notion of correlation surface as a
convenient tool for comparing portfolio credit loss generating models and
pricing synthetic CDO\ tranches. Analyzing alternative specifications of the
underlying dynamics, we conclude that the asymmetric models with TARCH
volatility specification are the preferred choice for generating significant
and persistent credit correlation skews. The characteristic dependence of the
correlation skew on term to maturity and portfolio hazard rate in these models
has a significant impact on both relative value analysis and risk management
of CDO tranches.

\end{abstract}

\section{Introduction}

The latest advances in credit correlation modeling were in part motivated by
the growth and sophistication of the so called correlation trading strategies,
in particular those involving the standard tranches referencing the Dow Jones
CDX (US) and iTraxx (Europe) broad market CDS indexes. The synthetic CDO
market allows investors to take views on the shape of the credit loss
distribution of the underlying collateral portfolio. The market implied
portfolio loss distribution is now well exposed through the pricing of liquid
standard tranches, which in turn are expressed through their implied correlations.

The pricing of credit derivatives has been based on either structural Merton
style models or reduced form credit migration models. See Lando \cite{Lando}
and Schonbucher \cite{SchonbucherBook} for a survey of these techniques. In
both cases, rather ad hoc models of dependence are needed to explain why
correlations vary over time and across tranches. Frequently, specific copulas
are postulated to model prices at a point in time. However, this approach does
not easily generalize to dynamic situations where new information is
continually being revealed and prices of tranches of different maturities
evolve consistently with each other and with the underlying reference portfolio.

This paper brings the standard time series methodology into the portfolio
credit risk setting. Specifically, we consider a structural credit model where
the latent asset variables evolve according to a one factor multi-variate
asymmetric GARCH model. This model is formulated and estimated on high
frequency data and the implications for long horizon loss distributions are
derived by time aggregation. We demonstrate that the asymmetric GARCH
specification can generate multivariate return distributions with both
significant lower tail dependence and asymmetry that persist for very long
time horizons. We then show that a multiname credit portfolio loss
distribution derived from our model can produce implied correlation skews
similar to the ones observed in the synthetic CDO market.

The article is organized as follows. In section \ref{section_portfolio_credit}
we give a brief overview of the portfolio credit modeling in a general copula
framework, and introduce the notion of correlation surface as a generalization
of base correlation to a dynamic setting. We derive analytical formulas
expressing the portfolio loss distribution directly in terms of the shape of
the correlation surface. In section \ref{section_equity_return} we apply time
series models to the portfolio credit risk problem. We develop new
implications of the well known asymmetric threshold GARCH (TARCH) models and
show their power in explaining persistent non-Gaussian features of long
horizon market returns. We then demonstrate the implications of using such
latent variable specification for default correlation estimates. In section
\ref{section_model_comparison} we use the correlation surface estimates to
assess the ability of various static and dynamic credit loss-generating models
to produce realistic pricing of CDO tranches. In section
\ref{section_conclusions} we summarize the results and outline possible
applications and extensions of our approach.

\section{\label{section_portfolio_credit}Defining Portfolio Credit Risk
Models}

In this section we will introduce the notion of correlation surface as a
general tool for comparing portfolio credit risk models. We will first give a
brief overview of the general copula framework, focusing particularly on the
symmetric one factor latent variable assumption and the large homogeneous
portfolio approximation which will be used throughout this paper. We will then
formalize the definition of the correlation surface and show that, under the
above approximation, it contains sufficient information to recover the full
portfolio loss distribution and to price CDO tranches.

\subsection{\label{section_copula}General Copula Framework}

Consider a portfolio of $M$ credit-risky obligors in a static setup with a
fixed time horizon $[0,T]$. To simplify notations we will skip the time
subscript for time dependent variables. At time $t=0$ all $M$ obligors are
assumed to be in non-default state and at time $T$ firm $i$ is in default with
probability $p_{i}.$ We assume we know the individual default probabilities
$\mathbf{p}=\left[  p_{1},...,p_{M}\right]  ^{\prime}$ (either risk-neutral,
e.g. inferred from default swap quotes, or actual, e.g. estimated by rating
agencies). Let $\tau_{i}\geq0$ be the random default time of obligor $i$ and
$Y_{i}=1_{\{\tau_{i}\leq T\}}$ the default dummy variable which is equal to 1
if default happened before $T$ and $0$ otherwise.

The loss generated by obligor $i$ conditional on its default is denoted as
$l_{i}>0$. The loss $l_{i}$ is a product of the total exposure $n_{i}$ and
percentage loss given default $1-\overset{\_}{R}_{i}$ where $\overset{\_}%
{R}_{i}\in\lbrack0,1]$ is the recovery rate. We assume that all $l_{i}$ are
constant (see \cite{Anderson-Sidenius} for discussion on stochastic
recoveries). Portfolio loss $L_{M}$ at time $T$ is the sum of the individual
losses for the defaulted obligors%
\begin{equation}
L_{M}=\sum_{i=1}^{M}l_{i}1_{\{\tau_{i}\leq T\}}=\sum_{i=1}^{M}l_{i}Y_{i}%
\end{equation}

The mean loss of the portfolio can be easily calculated in terms of individual
default probabilities:%
\begin{equation}
E\left(  L_{M}\right)  =\sum_{i=1}^{M}l_{i}E\left(  Y_{i}\right)  =\sum
_{i=1}^{M}l_{i}p_{i} \label{ELossPorfolio}%
\end{equation}

Risk management and pricing of derivatives contingent on the loss of the
credit portfolio, such as CDO\ tranches, require knowing not only the mean but
the whole distribution of losses with cdf $F_{L}\left(  x\right)  =P\left(
L_{M}\leq x\right)  .$ Portfolio loss distribution depends on the joint
distribution of default indicators $\mathbf{Y}=\left[  Y_{1},...,Y_{M}\right]
^{\prime}$ and in a static setup can be conveniently modeled using the latent
variables approach \cite{Frey-McNeil-2001}. Particularly, to impose structure
on the joint distribution of default indicators we assume that there exists a
vector of $M$ real-valued random variables $\mathbf{R}=\left[  R_{1}%
,...,R_{M}\right]  ^{\prime}$ and $M$ dimensional vector of non-random default
thresholds $\mathbf{d}=\left[  d_{1},...,d_{M}\right]  ^{\prime}$ such that
\begin{equation}
Y_{i}=1\Longleftrightarrow R_{i}\leq d_{i}\text{ for }i=1,...,M
\end{equation}

Denote $F:\mathbb{R}^{M}\rightarrow\lbrack0,1]$ as a cdf of $\mathbf{R}$ and
assume that it is a continuous function with marginal cdf $\left\{
F_{i}\right\}  _{i=1}^{M}$. For each obligor $i$ the default threshold $d_{i}$
is calibrated to match the obligor's default probability $p_{i}$ by inverting
the cdf of its aggregate returns $R_{i}:$ $d_{i}=F_{i}^{-1}\left(
p_{i}\right)  $. According to Sklar's theorem \cite{Sklar}, under the
continuity assumption $F$ can be uniquely decomposed into marginal cdfs
$\left\{  F_{i}\right\}  _{i=1}^{M}$ and the $M$-dimensional copula
$C:[0,1]^{M}\Rightarrow\lbrack0,1]$
\begin{equation}
F\left(  \mathbf{d}\right)  =C\left(  F_{1}\left(  p_{i}\right)
,..,F_{M}\left(  p_{M}\right)  \right)
\end{equation}

The most popular copula choices are the Gaussian copula model \cite{Li-2000},
Student-t \cite{Mashal-Naldi-Zeevi-2003}, and double-t \cite{HullWhite2004}.
The choice of copula $C$ defines the joint distribution of default indicators
from which the portfolio loss distribution can be calculated. The number of
names in the portfolio can be large and therefore the calibration of the
copula parameters can be problematic. To reduce the number of parameters some
form of symmetry is usually imposed on the distribution of default indicators.
Gordy \cite{GordyM} and Frey and McNeil \cite{Frey-McNeil-2001} discuss the
mathematics behind the modeling of credit risk in homogeneous groups of
obligors and the equivalence of the homogeneity assumption to the factor
structure of default generating variables.

\textbf{Assumption 1(Symmetric One Factor Model)}\emph{: Assume that loss
given default }$l_{i}=(1-\overset{\_}{R}_{i})\cdot n_{i}$\emph{ and individual
default probabilities }$p_{i}$\emph{ are the same for all }$M$\emph{ names in
the portfolio and that the latent variables admit symmetric linear one factor
representation:}%
\begin{subequations}
\begin{align}
n_{i}  &  =n\\
\overset{\_}{R}_{i}  &  =\overset{\_}{R}\\
p_{i}  &  =p\\
R_{i}  &  =bR_{m}+\sqrt{1-b^{2}}E_{i}\text{ with }0\leq b\leq1
\end{align}
\emph{where }$R_{m}$\emph{ and }$E_{i}$\textbf{ }\emph{are independent zero
mean, unit variance random variables. }$E_{i}^{\prime}$\emph{-s are
identically distributed with cdf }$G(\bullet)$\emph{.}

Within this framework, denote:
\end{subequations}
\begin{itemize}
\item $F\left(  d_{i}\right)  \equiv P\left(  R_{i}\leq d_{i}\right)  $ cdf of
aggregate total returns $R_{i}$

\item $G\left(  d_{i}\right)  \equiv P\left(  E_{i}\leq d_{i}\right)  $ cdf of
aggregate idiosyncratic returns $E_{i}$

\item $F\left(  \mathbf{d}\right)  \equiv P\left(  \mathbf{R}\leq
\mathbf{d}\right)  $ joint cdf of $\mathbf{R}$

\item $C\left(  \mathbf{u}\right)  \equiv F\left(  F^{-1}\left(  u_{1}\right)
,...,F^{-1}\left(  u_{M}\right)  \right)  $ copula of $\mathbf{R}$
\end{itemize}

Note than the assumption of one factor structure implies that equity returns
$\mathbf{R}$ are independent conditional on the market return $R_{m}$ and
therefore $F\left(  \mathbf{d}\right)  $ can be computed as expectation of the
product of conditional cdfs:%
\begin{equation}
F\left(  \mathbf{d}\right)  =E\left(
{\displaystyle\prod\limits_{i=1}^{M}}
P\left(  R_{i}\leq d_{i}|R_{m}\right)  \right)  =E\left(
{\displaystyle\prod\limits_{i=1}^{M}}
G\left(  d_{i}-b_{i}R_{m}\right)  \right)
\end{equation}

Parameter $b$ defines the pairwise correlation of latent variables with the
market factor. The correlation of latent variables, $\rho,$ which is often
referred to as "asset correlation" (this naming reflects the interpretation of
latent variables as asset returns in Merton-style structural default models),
is constant across all pairs of assets in a symmetric single factor model:%

\begin{equation}
\rho_{ij}=\rho=b^{2} \label{asset_corr}%
\end{equation}

In addition to the asset correlation, which reflects the co-movement of
returns on small scale, multivariate distributions can be also characterized
by measures that reflect joint extreme movements for a pair of assets -- the
tail dependence coefficient $\lambda_{ij}^{d}$ and the pairwise default
correlation coefficient $\rho_{ij}^{d}\left(  p\right)  $. Suppose $R_{i}$ and
$R_{j}$ are the stock returns for companies $i$ and $j$ over the $[0,T]$ time
horizon. The coefficient of lower tail dependence and the default correlation
coefficient for two random variables with the same continuous marginal cdfs,
$F\left(  R\right)  ,$ and the same default probabilities, $p,$are defined as:%
\begin{align}
\lambda_{ij}^{d}  &  =\lim_{p\rightarrow+0}P\left(  R_{i}\leq d_{p}|R_{j}\leq
d_{p}\right)  =\lim_{p\rightarrow+0}\frac{C\left(  p,p\right)  }%
{p}\label{tail_dependence_coeff}\\
\rho_{ij}^{d}\left(  p\right)   &  =corr(1_{\left\{  R_{i}\leq d_{p}\right\}
},1_{\left\{  R_{j}\leq d_{p}\right\}  })=\frac{C\left(  p,p\right)  -p^{2}%
}{(1-p)p} \label{pairwise_default_corr}%
\end{align}
where $p$ is the probability of crossing the threshold (also interpreted as
the default probability), and is related to the latter via the relationship
$d_{p}=F^{-1}\left(  p\right)  .$ Both these measures depend only on the
bivariate copula of the two random variables and are asymptotically equal:
$\underset{p\rightarrow+0}{\lim}\rho_{ij}^{d}\left(  p\right)  =\lambda
_{ij}^{d}$. Generally speaking, the measures of small-scale and extreme
co-movement of assets are independent of each other. For example, it is quite
possible to have $\rho_{ij}=0$ and $\lambda_{ij}^{d}\neq0$ and vice versa for
a non-Gaussian multi-variate distribution. Embrecht \textit{et al}
\cite{Embrechts-Lindskog-McNeil-2001} provide very detailed introduction to
the properties of those dependence measures.

To simplify the calculations even more, the large homogenous portfolio (LHP)
approximation is often used. Suppose that we increase the number of names in
the portfolio while keeping the total exposure size of the portfolio constant
so that $n_{i}=N/M$. Conditional on $R_{m}$ the loss of the portfolio contains
the mean of independent identically distributed random variables,
$L_{M}=\left(  1-\overset{\_}{R}\right)  N\frac{1}{M}\sum_{i=1}^{M}1_{\left\{
R_{i}\leq d\right\}  },$ which a.s. converges to its conditional expectation
as $M$ increases to infinity. We use $L$ without subscript to denote the
portfolio loss under LHP\ assumption.

\begin{proposition}
(\textbf{LHP Loss)} Under Assumption 1%
\begin{align}
L  &  \equiv\lim_{M\rightarrow\infty}\left[  \left(  1-\overset{\_}{R}\right)
N\frac{1}{M}\sum_{i=1}^{M}1_{\left\{  R_{i}\leq d\right\}  }\right]  =\left(
1-\overset{\_}{R}\right)  NP\left(  R_{i}\leq d|R_{m}\right) \label{LHP_L}\\
&  =\left(  1-\overset{\_}{R}\right)  NG\left(  \frac{d-bR_{m}}{\sqrt{1-b^{2}%
}}\right)  \text{ a.s. for any }R_{m}\in\text{supp}\left(  G\right) \nonumber
\end{align}

\end{proposition}

\begin{proof}
see proposition 4.5 in \cite{Frey-McNeil-2001}
\end{proof}

Based on (\ref{LHP_L}) cdf of $L$ can be expressed in terms of the cdf of
$R_{m}$%
\begin{align}
P\left(  L\leq l\right)   &  =P\left(  R_{m}\geq d_{1}\left(  l\right)
\right) \\
d_{1}\left(  l\right)   &  =\frac{d}{b}-\frac{\sqrt{1-b^{2}}}{b}G^{-1}\left(
\frac{l}{\left(  1-\overset{\_}{R}\right)  N}\right)
\end{align}

We use the following notation for the Gaussian distribution:

\begin{notation}
Let $\Phi\left(  .\right)  $ and $\phi(.)$ with one agrument denote the cdf
and pdf, correspondingly, of a standard normal random variable. Let
$\Phi\left(  .,.;\rho\right)  $ and $\phi(.,.;\rho)$ with three arguments
denote cdf and pdf of two standard normal random variables with linear
correlation $\rho$, and $\Phi_{i}\left(  .,.;\rho\right)  $ denote the partial
derivative of $\Phi\left(  .,.;\rho\right)  $ with respect to the i'th
agrument, e.g. $\Phi_{3}\left(  .,.;\rho\right)  \equiv\frac{\partial
}{\partial\rho}\Phi\left(  .,.;\rho\right)  .$
\end{notation}

For the Gaussian copula we have the familiar formula for LHP loss first
derived by Vasicek \cite{Vasicek-LHP}, where we have substituted the asset
correlation parameter $\rho$ in place of the factor loading $b$ using the
relation (\ref{asset_corr}):%
\begin{gather}
L^{G}=\left(  1-\overset{\_}{R}\right)  N\Phi\left(  \frac{\Phi^{-1}\left(
p\right)  -\sqrt{\rho}R_{m}}{\sqrt{1-\rho}}\right) \\
P\left(  L\leq l\right)  =1-\Phi\left(  d_{1}^{G}\left(  l\right)  \right) \\
d_{1}^{G}\left(  l\right)  =\frac{\Phi^{-1}\left(  p\right)  }{\sqrt{\rho}%
}-\frac{\sqrt{1-\rho}}{\sqrt{\rho}}\Phi^{-1}\left(  \frac{l}{\left(
1-\overset{\_}{R}\right)  N}\right)
\end{gather}

Vasicek \cite{Vasicek-LHP} and Schonbucher and Shubert
\cite{Schonbucher-Shubert-2001} show that LHP approximation is quite accurate
for upper tail of the loss distribution even for mid-sized portfolios of about
100 names. We will use a symmetric one factor LHP approximation in this paper
for analytical tractability.

\subsection{\label{section_correlation_spectrum}From Loss Distribution to
Correlation Surface}

It is intuitively clear that the choice of the dependence structure affects
the degree of uncertainty about the portfolio loss. Indeed, if all issuers in
the portfolio are completely independent of each other and any common driving
factor, then the law of large numbers assures that the portfolio loss under
the LHP approximation is a well determined number with little uncertainty
about it. On the other extreme, if the issuers are highly dependent such that
they all default or survive at the same time, then the portfolio loss has a
binary outcome -- it is either zero or equal to the maximum loss.

It is easy to quantify this intuitive result. While the mean of the loss
distribution is not affected by the choice of copula, one can show that the
second and higher moments of the loss distribution depend on the copula
characteristics. In particular, the variance of the loss can be expressed in
terms of bivariate default correlation coefficient $\rho^{d}\left(  p\right)
$ defined in section \ref{section_copula}. Under assumption of equal default
probabilities for all obligors, it is given by:%
\begin{equation}
Var(L)=(1-\overset{\_}{R})^{2}N^{2}p(1-p)\rho^{d}\left(  p\right)
\label{VarL}%
\end{equation}
Thus, in line with the intuition, the uncertainty of the default loss
distribution is directly proportional to the default correlation coefficient.

To characterize the shape of the entire distribution of $L$, we must look at
the particular slices of portfolio loss. Let $\left(  K_{d},K_{u}\right]  $
denote a tranche with attachment point $K_{d}$ and detachment point $K_{u}$
expressed as fractions of the reference portfolio notional so that $0\leq
K_{d}<K_{u}\leq1$. The notional of the tranche is defined as $N_{\left(
K_{d},K_{u}\right]  }=N\left(  K_{u}-K_{d}\right)  $ where $N$ is the notional
of the portfolio. The loss $L_{\left(  K_{d},K_{u}\right]  }$ of the tranche
is the fraction of $L$ that falls between $K_{d}$ and $K_{u}.$ For simplicity,
assume that total notional $N$ is normalized to 1.%

\begin{gather}
L_{\left(  K_{d},K_{u}\right]  }=f_{\left(  K_{d},K_{u}\right]  }\left(
L\right) \\
f_{\left(  K_{d},K_{u}\right]  }\left(  x\right)  \equiv\left(  x-K_{d}%
\right)  ^{_{+}}-\left(  x-K_{u}\right)  ^{_{+}} \label{fKK}%
\end{gather}

Tranches with zero attachment point, $\left(  0,K_{u}\right]  ,$ and unit
detachment point, $\left(  K_{d},1\right]  ,$ are called equity and senior
tranches, respectively. Loss of any tranche can be decomposed into losses of
either two equity or two senior tranches tranches $L_{\left(  K_{d}%
,K_{u}\right]  }=L_{\left(  0,K_{u}\right]  }-L_{\left(  0,K_{d}\right]
}=L_{\left(  K_{d},1\right]  }-L_{\left(  K_{u},1\right]  }$ This is similar
to representing a spread option as a long/short position using either calls or puts.

The expected loss of the equity tranche $L_{\left(  0,K\right]  }$ depends on
the portfolio loss distribution and under the LHP\ approximation can be
computed using only the distribution of the market factor:%

\begin{align}
EL_{\left(  0,K\right]  }  &  =Ef_{\left(  0,K\right]  }\left(  L\right)
\label{ELK}\\
&  =\left(  1-\overset{\_}{R}\right)  E\left[  G\left(  \frac{d-bR_{m}}%
{\sqrt{1-b^{2}}}\right)  1_{\left\{  R_{m}\geq d_{1}\left(  K\right)
\right\}  }\right]  +KP\left(  R_{m}<d_{1}\left(  K\right)  \right) \nonumber
\end{align}
The expectation in (\ref{ELK}) can be computed by Monte Carlo simulation or
numerical integration if we know the cdf of residuals $G$ and the distribution
of $R_{m}$ (see appendix \ref{appendix_LHP_MonteCarlo}). For the Gaussian
copula, the integral can be calculated in a closed form:%

\begin{align}
E^{G}L_{\left(  0,K\right]  }  &  =\left(  1-\overset{\_}{R}\right)
\Phi\left(  \Phi^{-1}\left(  p\right)  ,-d_{1};-\sqrt{\rho}\right)
+K\Phi\left(  d_{1}\right) \label{ELKG}\\
d_{1}  &  =\frac{1}{\sqrt{\rho}}\Phi^{-1}\left(  p\right)  -\frac{\sqrt
{1-\rho}}{\sqrt{\rho}}\Phi^{-1}\left(  \frac{K}{1-\overset{\_}{R}}\right)
\end{align}

Because of its analytical tractability, it is convenient to use the Gaussian
copula as a benchmark model when comparing different choices of dependence
structure. By finding the asset correlation level $\rho$ that replicates the
results of more complex portfolio loss generating models in the context of a
Gaussian copula framework, we can translate the salient features of such
models into mutually comparable units. More specifically, we define the
correlation surface as follows:

\begin{definition}
Suppose the loss distribution of a large homogeneous portfolio is generated by
a model $\left\{  C,p,\overset{\_}{R}\right\}  $ with copula $C,$ identical
individual default probabilities $p$ and recovery rates $\overset{\_}{R}$. Let
$L_{(0,K]}\in\left[  pf_{\left(  0,K\right]  }\left(  1-\overset{\_}%
{R}\right)  ,\text{ }f_{\left(  0,K\right]  }\left(  \left(  1-\overset{\_}%
{R}\right)  p\right)  \right]  $ be the expected loss of the equity tranche
$\left(  0,K\right]  .$ We define the \textbf{correlation surface}
$\rho(K,p,\overset{\_}{R})$ of the model $\left\{  C,p,\overset{\_}%
{R}\right\}  $ as the correlation parameter of the Gaussian copula that
produces the same expected loss $EL_{(0,K]}$ for the tranche $\left(
0,K\right]  $ for the given horizon $T$ and given single-issuer cumulative
default probability $p$:%
\begin{gather}
\rho(K,p,\overset{\_}{R})\text{ solves }E^{G}L_{(0,K]}\left(  \rho\right)
=EL_{(0,K]}\text{ for all }K\in\left[  0,1\right] \label{corr_imp}\\
\text{where }EL_{(0,K]}\text{ is expected loss of the tranche}\nonumber\\
(0,K]\text{ generated by model }\left\{  C,p,\overset{\_}{R}\right\} \nonumber
\end{gather}
where $E^{G}L_{(0,K]}$ is defined in (\ref{ELKG}).
\end{definition}

The correlation surface as defined above is closely related but not identical
to the notion of the base correlation used by many practitioners
\cite{JPMorganGuide}. The difference is that the base correlation is defined
using the prices of the equity tranches, which in turn depend on interest
rates, term structure of losses, etc. By contrast, the correlation surface is
defined without a reference to any market price. It characterizes the
portfolio loss generating model, rather than the supply/demand forces in the market.

One could, of course, take another logical step, and instead of deriving the
correlation surface from the parameters of the dynamic loss generating model,
go in the opposite direction -- find such parameters of the loss generating
model that result in the closest match to the market prices of CDO tranches.
It would be natural to call this solution "implied parameters", and the
corresponding function $\rho(K,p,\overset{\_}{R})$ "implied correlation
surface". The latter would, in fact, coincide with the conventionally defined
base correlation.

To ensure that the correlation surface is well defined we need to prove that
(\ref{corr_imp}) has a unique solution. Let us first prove the following:

\begin{proposition}
For the Gaussian copula, the expected loss of an equity tranche,
$E^{G}L_{\left(  0,K\right]  },$ is a monotonically decreasing function of
$\rho$ and attains its maximum when $\rho$ is equal to 0 and minimum when
$\rho$ is 1.

\begin{proof}
Using Notation 2, Eq. (\ref{ELKG}) and the properties of Gaussian
distribution\footnote{The following properties of two dimensional Gaussian cdf
are used in the calculation%
\begin{gather*}
\Phi_{2}\left(  x,y;\rho\right)  \equiv\frac{\partial}{\partial y}\Phi\left(
x,y;\rho\right)  =\phi\left(  y\right)  \Phi\left(  \frac{x-\rho y}%
{\sqrt{1-\rho^{2}}}\right) \\
\Phi_{3}\left(  x,y;\rho\right)  \equiv\frac{\partial}{\partial\rho}%
\Phi\left(  x,y;\rho\right)  =\phi\left(  x,y;\rho\right)
\end{gather*}
\par
First formula is derived by taking the derivative and re-arranging the terms.
The proof of the second can be found in
Vasicek(\cite{Vasicek-BiNormalExpansion})
\par
{}} we derive
\begin{align}
E_{\rho}^{G}L_{\left(  0,K\right]  }  &  \equiv\frac{\partial}{\partial\rho
}E^{G}L_{\left(  0,K\right]  }\\
&  =-\left(  1-\overset{\_}{R}\right)  \left[  \frac{1}{2\sqrt{\rho}}\Phi
_{3}\left(  \Phi^{-1}\left(  p\right)  ,-d_{1};-\sqrt{\rho}\right)  +\Phi
_{2}\left(  \Phi^{-1}\left(  p\right)  ,-d_{1};-\sqrt{\rho}\right)
\frac{\partial}{\partial b}d_{1}\right] \nonumber\\
&  +K\phi\left(  d_{1}\right)  \frac{\partial}{\partial b}d_{1}\nonumber\\
&  =-\frac{1-\overset{\_}{R}}{2\sqrt{\rho}}\phi\left(  \Phi^{-1}\left(
p\right)  ,-d_{1};-\sqrt{\rho}\right)  <0\text{ }%
\end{align}
\textit{for any }$\rho\in\left(  0,1\right)  $\textit{.}
\end{proof}
\end{proposition}

Therefore, there is a one-to-one mapping between loss distribution and
correlation surface, and our transformation does not lead to any loss of
information. The next proposition shows how to calculate the loss cdf using
the correlation surface and its slope along the $K$-dimention.

\begin{proposition}
Suppose $\rho(K,p,\overset{\_}{R})$ is the correlation surface for model
$\left\{  C,p,\overset{\_}{R}\right\}  $ and the probability distribution
function of the portfolio loss is a continuous function then the loss cdf can
be computed from the correlation surface:%
\begin{equation}
P\left(  L\leq K\right)  =P^{G}\left(  L\leq K\right)  +\rho_{K}E_{\rho}%
^{G}L_{\left(  0,K\right]  } \label{loss_from_corr}%
\end{equation}
where $\rho\equiv\rho(K,p,\overset{\_}{R})$ is the level of the correlation
surface, $\rho_{K}\equiv\frac{\partial}{\partial K}\rho(K,p,\overset{\_}{R})$
is the correlation surface slope and%
\begin{gather}
P^{G}\left(  L\leq K\right)  =1-\Phi\left(  d_{1}\right) \\
E_{\rho}^{G}L_{\left(  0,K\right]  }=\left(  1-\overset{\_}{R}\right)
\frac{1}{2\sqrt{\rho}}\phi\left(  \Phi^{-1}\left(  p\right)  ,-d_{1}%
;-\sqrt{\rho}\right) \\
d_{1}=\frac{1}{\sqrt{\rho}}\Phi^{-1}\left(  p\right)  -\frac{\sqrt{1-\rho}%
}{\sqrt{\rho}}\Phi^{-1}\left(  \frac{K}{1-\overset{\_}{R}}\right)
\end{gather}

\begin{proof}
first note that the derivative with respect to $K$ of the expected tranche's
loss under true copula $C$ is related to the cdf of the loss%
\begin{align}
\frac{d}{dK}EL_{\left(  0,K\right]  }  &  =\frac{d}{dK}E\left(  L-\left(
L-K\right)  _{+}\right) \\
&  =-\frac{d}{dK}E\left(  L-K\right)  _{+}=E1_{\left\{  L-K\geq0\right\}
}=1-P\left(  L\leq K\right) \nonumber
\end{align}
therefore%
\begin{align}
P\left(  L\leq K\right)   &  =1-\frac{d}{dK}EL_{\left(  0,K\right]  }\\
&  =1-E_{K}^{G}L_{\left(  0,K\right]  }-\rho_{K}E_{\rho}^{G}L_{\left(
0,K\right]  }\nonumber\\
&  =1-\Phi\left(  d_{1}\right)  +\frac{1-\overset{\_}{R}}{2\sqrt{\rho}}%
\phi\left(  \Phi^{-1}\left(  p\right)  ,-d_{1};-\sqrt{\rho}\right)  \rho
_{K}\nonumber\\
&  =P^{G}\left(  L\leq K\right)  +\left(  1-\overset{\_}{R}\right)  \frac
{1}{2\sqrt{\rho}}\phi\left(  \Phi^{-1}\left(  p\right)  ,-d_{1};-\sqrt{\rho
}\right)  \rho_{K}\nonumber
\end{align}
where partial derivative with respect to $K$ is computed as%
\begin{align}
E_{K}^{G}L_{\left(  0,K\right]  }  &  \equiv\frac{\partial}{\partial K}%
E^{G}L_{\left(  0,K\right]  }\\
&  =-\left(  1-\overset{\_}{R}\right)  \Phi_{2}\left(  \Phi^{-1}\left(
p\right)  ,-d_{1};-\sqrt{\rho}\right)  \frac{\partial}{\partial K}d_{1}%
+K\phi\left(  d_{1}\right)  \frac{\partial}{\partial K}d_{1}+\Phi\left(
d_{1}\right) \nonumber\\
&  =\Phi\left(  d_{1}\right) \nonumber
\end{align}

\end{proof}
\end{proposition}

Another important point is that the correlation surface depends implicitly on
the term to maturity via the cumulative default probability $p$. However, this
is not the only dependence -- potentially, the dependence structure
characterized by the copula $C$ also exhibits some time dependence when viewed
within the context of the Gaussian copula. This statement needs a
clarification -- the copula $C$ itself is defined in a manner that encompasses
all time horizons and therefore cannot depend on any particular horizon.
However, when we translate the tranche loss generated with this dependence
structure into the simpler Gaussian model the transformation that is required
may depend on the horizon $T$. In other words, the shape of the correlation
surface $\rho(K,p,\overset{\_}{R})$ may depend on the horizon. We study this
dependence in detail in section \ref{section_model_comparison}.

These results allow us to further develop the often mentioned analogy between
the role that default correlation plays in tranche pricing on one hand and the
role that implied volatility plays in equity derivatives pricing, on the
other. The starting point of this analogy is the result (\ref{VarL}), which
shows that default correlation is directly related to the uncertainty of
portfolio loss. If we recall that the CDO equity tranche is essentially a call
option on portfolio survival, it becomes clear that the price of the equity
tranche should be positively related to the default (asset) correlation, just
as the price of any equity option is positively related to the implied
volatility of its underlying stock.

The correlation surface takes this analogy one step further. Just as the
implied volatility surface is sufficient to derive the risk-neutral
distribution of the underlying stock (see \cite{Derman-Kani-1994},
\cite{Dupire-1994}, \cite{Rubinstein-1994}, \cite{Jackwerth-Rubinstein-1996})
and price any European option, we have shown that the correlation surface is
sufficient to derive the full portfolio loss distribution and price any CDO tranche.

\section{\label{section_equity_return}Time Series Approach to Tail Risk}

\subsection{Motivation for the time series approach}

The analogy between the correlation surface and implied volatility surface
introduced in the previous section leads to further insight about the origins
of the large credit portfolio loss risks. Both equity and credit derivatives
pricing exhibits substantial deviations from the simplest Gaussian models of
the underlying assets. In particular, the equity index options implied
volatility exhibits a steep downward skew of implied volatility, and the CDO
tranches exhibit a steep upward skew of implied (base) correlations. To
further underline the similarities between these skews, note that a senior CDO
tranche with an attachment point that is higher than the expected loss on the
underlying portfolio can be thought as an out-of-money put option on portfolio
losses. When looked from this angle, the equity implied volatility skew and
credit correlation skew are tilted the same way -- towards the farther
out-of-money options.

Recall that the empirical distribution of returns does indeed exhibit
significant downside tails, and that a large part of the implied volatility
skew can be explained by the properties of the empirical distribution
\cite{Day-Lewis-1992}. Given the above mentioned analogies between the
synthetic CDO tranches and equity index options, it is quite natural to look
for a similar explanation of the implied correlation skew.

The standard Gaussian copula framework implicitly relies on the Merton-style
structural model for definition of default correlations. Therefore, if we are
to give \ an empirical explanation to the observed base correlation skew we
must start by giving an empirical meaning to the variables in this model. Our
working hypothesis in this paper will be that the meaning of the "market
factor" in the factor copula framework is the same as the market factor used
in the equity return modeling. As such, it is often possible to use an
observable broad market index such as S\&P 500 as a proxy for the economic
market factor, with an added convenience that there exists a long historical
dataset for its returns and a rich set of equity options data from which one
can glean independent information about their implied return distribution.

This hypothesis is not uncommon in portfolio credit risk modeling -- for
example, the authors of \cite{Mashal-Naldi-Zeevi-2003}\ emphasized the
importance of using a fat-tailed distribution of asset returns in the copula
framework in part by citing the empirical evidence from equity markets.
However, most researchers have focused on the single-period return
distribution characteristics.

In contrast, we focus in this paper on the long-run cumulative returns, and
prove that their distribution is quite distinct from that of the short-term
(single-period) returns. As we will show in the rest of this paper, it is the
time aggregation properties and the compounding of the asymmetric volatility
responses that make it possible to explain the credit correlation skew for 5-
or even 10-year horizons. Moreover, this dynamic explanation of the skew
allows one to make rather specific predictions for the dependence of this skew
on both the term to maturity and on the hazard rates and other model parameters.

We first specify time series properties of stock returns for high frequency
time intervals (daily or weekly) and then derive the distribution of stock
prices over longer horizons measured in months or even years. Assuming
geometric Brownian motion of stock prices on short intervals leads to the same
log-normal shape for the distribution of stock prices for all future horizons.
Models with more realistic dynamics can lead to richer distributions of the
time aggregated returns even if the high frequency shocks are Gaussian.

In the GARCH family of models \cite{EngleArch}, \cite{BollerslevGARCH},
\cite{Engle-Bollerslev-Nelson}, there have been many investigations of the
difference between the conditional and the unconditional distributions. These
models reveal an important explanation for the excess kurtosis in financial
returns but generally show no reason to expect skewness in returns. Diebold
\cite{Diebold} investigates the implications of time aggregation of GARCH
models concluding that eventually they lose their excess kurtosis as the
central limit theorem leaves the average distribution Gaussian.

Asymmetric volatility models were introduced by Nelson \cite{Nelson} and
further investigated by Glosten, Jaganathan and Runkle \cite{GJR} and by
Zakoian \cite{Zakoian}. The model to be discussed in this paper is essentially
the GJR model and will be called the Threshold-ARCH or TARCH. These models all
show that negative returns forecast higher volatility than positive returns of
the same magnitude. This observation is very widespread and is sometimes
called the leverage effect following Black \cite{Black}. However, it is most
likely due to risk aversion as in Campbell and Hentschell \cite{Campb} and in
this paper we will refer to it simply as asymmetric volatility.

The first indication that skewness could arise from time aggregation was
presented in Engle \cite{EngleSkew}. He showed that with asymmetric
volatility, the skewness of time aggregated returns could be more negative
than the skewness of the individual innovations. In the next section we show
analytically how this negative skewness depends upon time aggregation. We then
examine the multivariate distribution when the single common factor has TARCH dynamics.

\subsection{Univariate model: TARCH(1,1)}

Let $r_{t}$ be the log-return of a particular stock or an index such as SP500
from time $t-1$ to time $t$ . $\digamma_{t}$ denotes the information set
containing realized values of all the relevant variables up to time $t$. We
will use the expectation sign with subscript $t$ to denote the expectation
conditional on time $t$ information set: $E_{t}\left(  .\right)  =E\left(
.|\digamma_{t}\right)  .$ The time step that we use in the empirical part is 1
day or 1 week. Predictability of stock returns is negligible over such time
horizons and therefore we assume that the conditional mean is constant and
equal to zero:%

\begin{equation}
m_{t}\equiv E_{t-1}(r_{t})=0
\end{equation}
The conditional volatility $\sigma_{t}^{2}\equiv E_{t-1}(r_{t}^{2})$ of
$r_{t}$ in TARCH(1,1) has the autoregressive functional form similar to the
standard GARCH(1,1) but with an additional asymmetric term \cite{GJR},
\cite{Zakoian}:%
\begin{align}
r_{t}  &  =\sigma_{t}\varepsilon_{t}\label{tarch11}\\
\sigma_{t}^{2}  &  =\omega+\alpha r_{t-1}^{2}+\alpha_{d}r_{t-1}^{2}1_{\left\{
r_{t-1}\leq0\right\}  }+\beta\sigma_{t-1}^{2}\nonumber
\end{align}

where $\left\{  \varepsilon_{t}\right\}  $ are iid , have zero mean, variance
normalized to 1, finite skewness $s_{\varepsilon}$ and finite kurtosis
$k_{\varepsilon}$. We also assume that $\omega>0$ and $\alpha,\alpha_{d}%
,\beta$ are non-negative so that the conditional variance $\sigma_{t}^{2}$ is
guaranteed to be positive.

The persistence of volatility in the model is governed by the parameter
$\zeta:$%

\begin{equation}
\zeta\equiv E\left(  \beta+\alpha\varepsilon_{t}^{2}+\alpha_{d}\varepsilon
_{t}^{2}1_{\left\{  \varepsilon_{t}\leq0\right\}  }\right)  =\beta
+\alpha+\alpha_{d}v_{\varepsilon}^{d} \label{tarch11_persistence}%
\end{equation}

where $v_{\varepsilon}^{d}\equiv E\left(  \varepsilon_{t}^{2}1_{\left\{
\varepsilon_{t}\leq0\right\}  }\right)  $ is the "right truncated variance" of
$\varepsilon_{t}$ $.$ If $\zeta\in\lbrack0,1)$ then conditional variance
mean-reverts to its unconditional level $\sigma^{2}=E\left(  \sigma_{t}%
^{2}\right)  =\frac{\omega}{1-\zeta}$. The following parameter $\xi$ will also
be useful in describing the higher moments of TARCH(1,1) returns and volatilities:%

\begin{equation}
\xi\equiv E\left(  \beta+\alpha\varepsilon_{t}^{2}+\alpha_{d}\varepsilon
_{t}^{2}1_{\left\{  \varepsilon_{t}\leq0\right\}  }\right)  ^{2}=\beta
^{2}+\alpha^{2}k_{\varepsilon}+\alpha_{d}^{2}k_{\varepsilon}^{d}+2\alpha
\beta+2\alpha_{d}\beta v_{\varepsilon}^{d}+2\alpha\alpha_{d}k_{\varepsilon
}^{d} \label{tarch11_xi}%
\end{equation}

where $k_{\varepsilon}^{d}\equiv E\left(  \varepsilon_{t}^{4}1_{\left\{
\varepsilon_{t}\leq0\right\}  }\right)  $ is the "right truncated kurtosis" of
$\varepsilon_{t}.$ We can rewrite (\ref{tarch11}) in terms of the increments
of the conditional volatility $\Delta\sigma_{t+1}^{2}\equiv\sigma_{t+1}%
^{2}-\sigma_{t}^{2}$ and the volatility shocks $\eta_{t}$%

\begin{align}
r_{t}  &  =\sigma_{t}\varepsilon_{t}\label{tarch11_v}\\
\Delta\sigma_{t+1}^{2}  &  =\left(  1-\zeta\right)  \left(  \sigma^{2}%
-\sigma_{t}^{2}\right)  +\sigma_{t}^{2}\eta_{t}\nonumber\\
\eta_{t}  &  \equiv\alpha\left(  \varepsilon_{t}^{2}-1\right)  +\alpha
_{d}\left(  \varepsilon_{t}^{2}1_{\left\{  \varepsilon_{t}\leq0\right\}
}-v_{\varepsilon}^{d}\right) \nonumber
\end{align}

The speed of mean reversion in volatility is $1-\zeta$ and is small when
$\zeta$ is close to one which is usually true for daily and weekly equity
returns -- hence the persistence of the volatility. The TARCH(1,1) volatility
shocks $\eta_{t}$ are iid, with zero mean and constant variance $var(\eta
_{t})=var(\alpha\varepsilon_{t}^{2}+\alpha_{d}\varepsilon_{t}^{2}1_{\left\{
\varepsilon_{t}\leq0\right\}  })=\xi-\zeta^{2}$. The correlation of
conditional volatility with the return in the previous period depends on the
covariance of return and volatility innovations:
\begin{equation}
corr_{t-1}\left(  r_{t},\sigma_{t+1}^{2}\right)  =corr_{t-1}\left(
\varepsilon_{t},\eta_{t}\right)  =\frac{\alpha s_{\varepsilon}+\alpha
_{d}s_{\varepsilon}^{d}}{\sqrt{\xi-\zeta^{2}}} \label{eqTarchLeverage}%
\end{equation}

where $s_{\varepsilon}^{d}=E\left(  \varepsilon_{t}^{3}1_{\left\{
\varepsilon_{t}\leq0\right\}  }\right)  <0$ is the "right truncated" skewness
of $\varepsilon_{t}.$The negative correlation of return and volatility shocks,
often cited as the "leverage effect"\footnote{Though we note here that the
magnitudfe of this "leverage effect" in return time series for stocks of most
investment grade issuers far exceeds the amount that would be reasonable based
purely on their capital structure leverage.}, is the main source of the
asymmetry in the return distribution. We can see from formula
(\ref{eqTarchLeverage}) that negative return-volatility correlation can be
achieved either through negative skewness of return innovations
$s_{\varepsilon}<0$, through asymmetry in volatility process $\alpha_{d}>0$ or
combination of the two$.$ We call these static and dynamic asymmetry, respectively.

In this paper we are interested in the effects of the volatility dynamics on
the distribution of long horizon returns which in the log representation is a
sum of short term log returns $R_{t,t+T}\equiv\ln S_{t+T}-\ln S_{t}=%
{\displaystyle\sum\limits_{u=t+1}^{t+T}}
r_{u}$. While a closed form solution for the probability density function of
TARCH(1,1) aggregated returns is not available, we can still derive some
analytical results for its conditional and unconditional moments: volatility,
skewness and kurtosis. Since TARCH is linear autoregressive volatility
process, the conditional variance $V_{t,t+T}$ of the log return $R_{t,t+T}$ is
linear in $\sigma_{t+1}^{2}$ :%

\begin{equation}
V_{t,t+T}=E_{t}R_{t,t+T}^{2}=E_{t}\left(
{\displaystyle\sum\limits_{t+1\leq u\leq t+T}}
\sigma_{u}^{2}\right)  =T\left(  \sigma^{2}+\left(  \sigma_{t+1}^{2}%
-\sigma^{2}\right)  \frac{1}{T}\frac{1-\zeta^{T}}{1-\zeta}\right)
\end{equation}

The new result of this paper for the TARCH(1,1) model is the representation of
the skewness term structure derived in the following proposition.

\begin{proposition}
Suppose $0\leq\zeta<1$ and the return innovations have finite skewness,
$s_{\varepsilon},$ and finite "truncated" third moment, $s_{\varepsilon}^{d}$.
Then the conditional third moment of $R_{t,t+T}$ has the following
representation for TARCH(1,1)
\begin{equation}
E_{t}R_{t,t+T}^{3}=s_{\varepsilon}%
{\displaystyle\sum\limits_{u=1}^{T}}
E_{t}\left(  \sigma_{t+u}^{3}\right)  +3\left(  \alpha s_{\varepsilon}%
+\alpha_{d}s_{\varepsilon}^{d}\right)
{\displaystyle\sum\limits_{u=1}^{T}}
\frac{1-\zeta^{T-u}}{1-\zeta}E_{t}\left(  \sigma_{t+u}^{3}\right)
\label{cond third moment}%
\end{equation}
In addition, if $E\sigma_{t}^{3}$ is finite, then unconditional skewness of
$R_{t,t+T}$ is given by
\begin{equation}
S_{T}\equiv\frac{ER_{t,t+T}^{3}}{E(R_{t,t+T}^{2})^{3/2}}=\left[  \frac
{1}{T^{1/2}}s_{\varepsilon}+3\frac{1}{T^{3/2}}\left(  \alpha s_{\varepsilon
}+\alpha_{d}s_{\varepsilon}^{d}\right)  \frac{T(1-\zeta)-1+\zeta^{T}}%
{(1-\zeta)^{2}}\right]  E\left(  \frac{\sigma_{t}}{\sigma}\right)  ^{3}
\label{skewness}%
\end{equation}

\begin{proof}
See appendix \ref{appendix_aggregate_TARCH} for the details.
\end{proof}
\end{proposition}

The conditional third moment is a function of the conditional term structure
of $\sigma_{t}^{3}$, term horizon $T$ and volatility parameters. The
conditional skewness can be computed using second and third conditional
moments derived above. The asymmetry in the return distribution arises from
two sources - skewness of return innovations and asymmetry of the volatility
process. Note that the second term in the formulas for conditional and
unconditional skewness is directly related to the correlation of return and
volatility innovations. If return-volatility correlation is zero $\left(
\alpha s_{\varepsilon}+\alpha_{d}s_{\varepsilon}^{d}=0\right)  $ then
$S_{T}=\frac{1}{T^{1/2}}s_{\varepsilon}E\left(  \frac{\sigma_{t}}{\sigma
}\right)  ^{3}$. If return innovations are symmetric then asymmetric
volatility drives the asymmetry in the return distribution. In figure
\ref{CondSkew} we show conditional and unconditional skewness term structures.
For realistic parameters corresponding approximately to parameters of the
TARCH(1,1) estimated for weekly SP500 log returns, both conditional and
unconditional skewness is negative. It decreases in the medium term, attains
the minimum at approximately the 2 year point and then decays to zero as T
increases. The skewness term structure conditional on the high/low current
volatility is above/below the unconditional skewness.%

\begin{figure}
[t]
\begin{center}
\includegraphics
[natheight=170.000000pt,natwidth=253.937500pt,height=172.9375pt,width=257.3125pt]%
{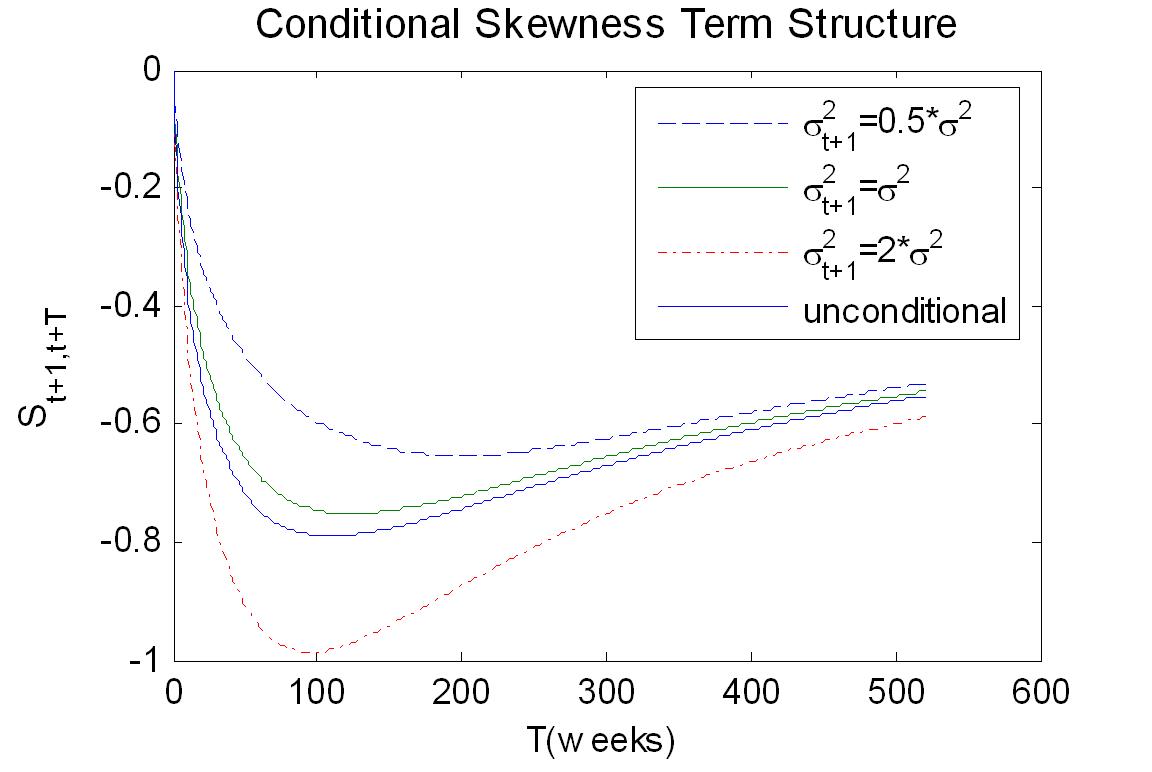}%
\caption{{\small Term structure of conditional skewness of time aggregated
return }$R_{t+1,t+T}.${\small TARCH(1,1) has persistence coefficient }%
$\zeta=0.98$ {\small and the following parametrization: }$\sigma^{2}=1,$
$\alpha=0.01,$ $\alpha_{d}=0.10,$ $\beta=0.92,$ $\varepsilon_{t}\sim N(0,1).$
{\small We plotted unconditional skewness term structure and conditional for
three different initial volatilities: }$\sigma^{2}/2${\small , }$\sigma^{2}%
${\small and }$2\sigma^{2}.$ {\small The term structure of }$E_{t}\sigma
_{t+u}^{3}${\small was computed from 10,000 independent simulations.}}%
\label{CondSkew}%
\end{center}
\end{figure}

To provide some empirical context to the theoretical discussion above, let us
consider the time series of SP500 returns. The results of the estimation of
various TARCH(1,1) specification are shown in the appendix
\ref{appendix_SP500}. Figure \ref{skewSP500} shows the estimate of skewness
for overlapping returns of different aggregation horizons measured in days.
The full sample shows high negative skewness for one day return because of the
1987 crash. On the post 1990 sample negative skewness rises with aggregation
horizon up to 1 year and then slowly decays toward zero. Both samples show
significant skewness for horizons of several years. We should note that
confidence bounds around skewness curves are quite wide due to the persistence
and high volatility of the squared returns and serial correlation of the
overlapping observations. Nevertheless, both the shape and the level of the
empirical skewness in figure \ref{skewSP500} and the theoretical estimate
shown in figure \ref{CondSkew} are quite similar.%
\begin{figure}
[t]
\begin{center}
\includegraphics
[trim=-1.272384pt 6.377927pt 1.272386pt -6.377914pt,
natheight=256.750000pt,natwidth=361.062500pt,height=172.875pt,width=242.625pt]%
{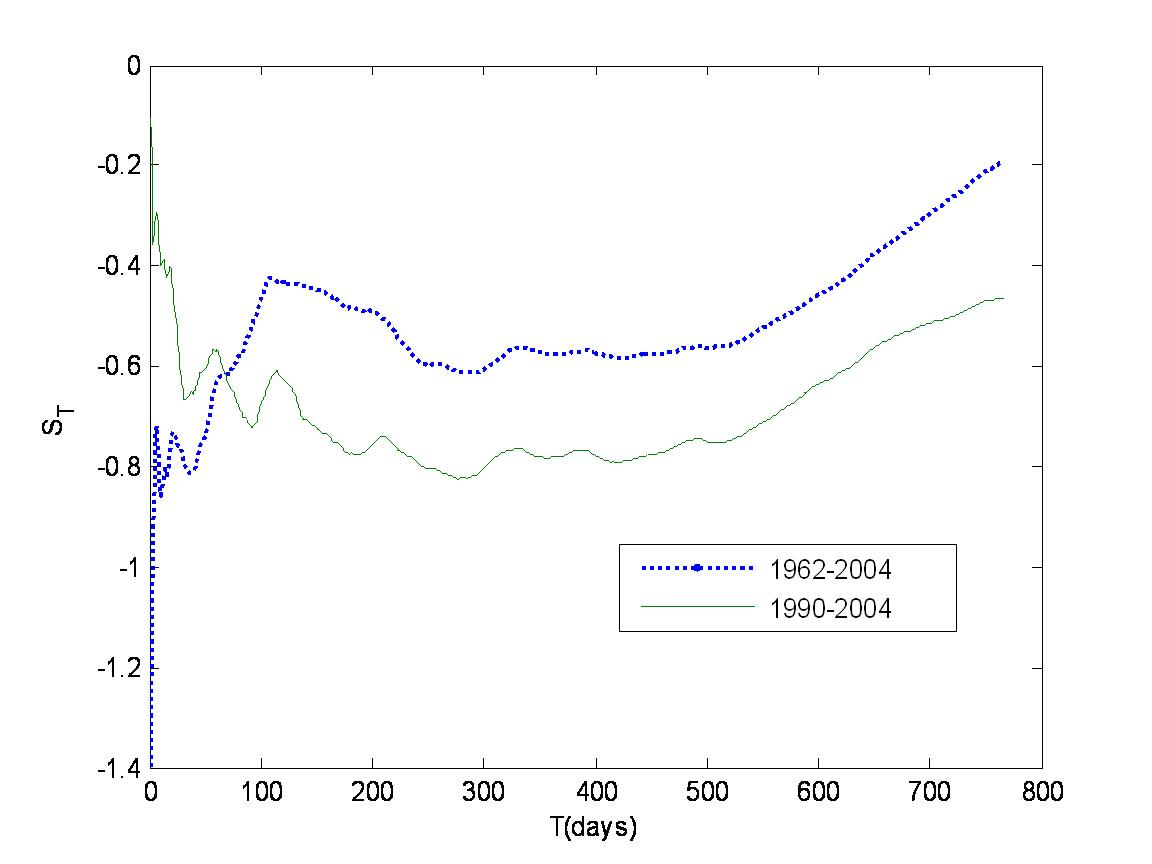}%
\caption{{\small Term structure of skewness for SP500 \ time aggregated log
returns estimated with overlapping samples moments for full and post-1990
data.}}%
\label{skewSP500}%
\end{center}
\end{figure}

\subsection{\label{SectionMultModel}Multivariate model\ with TARCH(1,1) factor
volatility dynamics}

Let us now turn to a multi-variate model of equity returns for $M$ companies,
with a simple dynamic factor structure decomposing the returns into a common
(market) and idiosyncratic components. To concentrate on the time dimension of
the model we assume a homogeneity of cross-sectional return properties, namely
that factor loadings and volatilities of idiosyncratic terms are constant and
identical for all stocks. Thus, our homogeneous one factor ARCH\ model has the
following form.%
\[
r_{i,t}=br_{m,t}+\sigma\varepsilon_{i,t}%
\]
where

\begin{itemize}
\item $b\geq0$ is the constant market factor loading and it is the same for
all stocks

\item $r_{m,t}$ is the market factor return with zero conditional mean
$E_{t-1}(r_{m,t})=0$, conditional volatility $\sigma_{m,t}^{2}\equiv
E_{t-1}(r_{m,t}^{2})$ that has TARCH(1,1) parametrization (\ref{tarch11})

\item $\sigma\varepsilon_{i,t}$ are the idiosyncratic return components with
constant volatilities $\sigma^{2}$ and zero conditional means $E_{t-1}%
(\sigma\varepsilon_{i,t})=0$

\item $\left\{  \varepsilon_{i,t},\varepsilon_{m,t}\right\}  $ are unit
variance iid shocks for each $t$ and all $i$
\end{itemize}

Because of the simple linear factor structure and constant market loadings
time aggregated equity returns $R_{i,T}=%
{\displaystyle\sum\limits_{u=1}^{T}}
r_{i,u}$ also have a one factor representation\footnote{To simplify the
notations we assume that the initial time $t=0$ and use only subscipt for the
time aggregation horizon $T.$}%
\begin{equation}
R_{i,T}=bR_{m,T}+E_{i,T}%
\end{equation}
where $R_{m,T}=%
{\displaystyle\sum\limits_{u=1}^{T}}
r_{m,u}$ and $E_{i,T}=\sigma%
{\displaystyle\sum\limits_{u=1}^{T}}
\varepsilon_{i,u}$ are independent conditional on $\digamma_{0}$.

\subsection{\label{section_tail_risk_default_corr}Tail risk and default
correlation estimates}

Assuming that the latent variables in the copula framework follow the one
factor TARCH(1,1) dynamics, we can calculate the coefficient of lower tail
dependence $\lambda_{i,j}^{d}$\ and the default correlation coefficient
$\rho_{i,j}^{d}\left(  p\right)  $, which\ were defined in equations
(\ref{tail_dependence_coeff}) and (\ref{pairwise_default_corr}), respectively.
Default correlation coefficient for Gaussian and Student-t can be computed in
closed form, while the factor GARCH/TARCH\ models require Monte Carlo
simulation, which is desribed in appendix C. Figure \ref{Figure_DefCorrBnds}
shows the numerical estimates of the default correlation $\rho_{1,2}^{d}$ as a
function of $p$ for 4 different models - TARCH, GARCH and Student-t and
Gaussian copulae.

The linear correlation of latent returns is set to 0.3 for all 3 models.
TARCH\ and GARCH\ are calibrated to have volatility dynamics parameters
corresponding approximately to the weekly SP500 returns and the time
aggregation horizon is set to 5 years. The degrees of freedom parameter for
the Student-t idiosyncracies and T-Copula is set to be equal to 12.

The GARCH distribution is symmetric and has smaller tail dependence for both
upper and lower tails. We can see on the graph that it also has lowest default
correlation for all default probabilities in the range of [0.01,0.2]. The
Student-t copula is also symmetric but has fatter joint tails compared to the
GARCH. Its default correlation is above GARCH for all p and converges to a
positive number (the tail dependence coefficient) as p decreases to zero.

We can see that TARCH\ has higher default correlation than other 2 models and
is upward sloping for very low quantiles. The upturn for the extreme tails is
a consequence of the left tail shape of the common factor. The default
correlation for very low default probabilities should be close to 1 since the
left tail of the factor is fatter than the left tail of the idiosyncratic
shocks. As we showed in the previous sections, both kurtosis and skewness of
the market factor declines faster for GARCH\ than TARCH given the same level
of volatility persistence.

We also show the $95\%$ confidence bounds for default correlation, which are
calculated using 1000 independent repetitions of 10,000 Monte Carlo
simulations of the common factor. TARCH, T-Copula and GARCH models have
virtually non-overlapping confidence bounds for $\rho_{i,j}^{d}\left(
p\right)  $ for sufficiently low values of the default probability [0.01,0.1].
Since the correlation skew for sufficiently large values of the detachment
point $K$ is also related to the limit of low default probabilities, we can
therefore rely on this result in the next section to assert that the
differences between the correlation surfaces generated by these three models
are statistically significant.%

\begin{figure}
[t]
\begin{center}
\includegraphics
[natheight=266.125000pt,natwidth=512.000000pt,height=163.5pt,width=313.25pt]%
{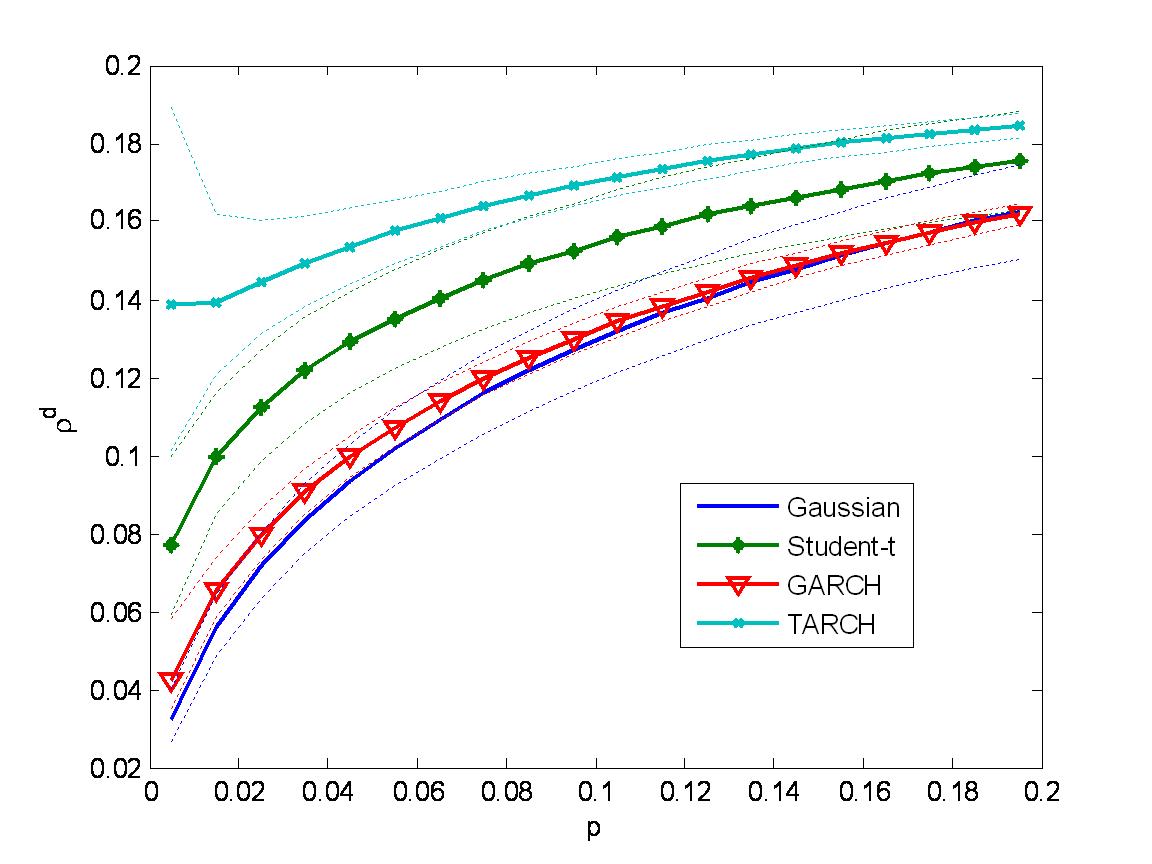}%
\caption{{\small Default correlation as a function of }$p$ {\small for
Gaussian copula, T-Copula, GARCH and TARCH models. The linear correlation
parameter is 0.3 for all four models. TARCH(GARCH) parameters are}%
${\small \ \ \alpha=\ 0.01(0.06),\ \alpha}_{d}{\small =\ 0.1(0),\ \beta
\ =\ 0.92(0.92)}${\small . The degrees of freedom parameter for T-Copula and
TARCH Student-t idiosyncracies is 12.}}%
\label{Figure_DefCorrBnds}%
\end{center}
\end{figure}

\section{\label{section_model_comparison}Comparing Credit Portfolio Loss
Generating Models}

Our goal in this paper is to provide a general framework for judging the
versatility of various portfolio credit risk models. All such models, whether
defined via dynamic multivariate returns model like in this paper or in
various versions of the static copula framework (\cite{Anderson-Sidenius},
\cite{Frey-McNeil-2001}, \cite{Gregory-Laurent-2003}, \cite{Li-2000},
\cite{Mashal-Naldi-Zeevi-2003} and \cite{Schonbucher-Shubert-2001}), can be
characterized by the full term structure of loss distributions. Thus, without
loss of generality, we can refer to all models of credit risk as loss
generating models, with an implicit assumption that any two models that
produce identical loss distributions for all terms to maturity are considered
to be equivalent.

The correlation surface, introduced in section
\ref{section_correlation_spectrum}, conveniently transforms specific choice of
a loss generating model into a two dimensional surface $\rho(K,T)$ of the
Gaussian copula correlation parameter, with the main dimensions being the loss
threshold (detachment level) $K$ and the term to maturity $T$. All other
inputs such as the recovery rate $R$, the term structure of (static) hazard
rates $h$, the level of linear asset correlation $\zeta$, the Student-t
degrees of freedom $\nu$, various GARCH model coefficients, etc. -- are
considered as model parameters upon which the two-dimensional correlation
surface itself depends.

Note that in the previous sections we have expressed the correlation surface
as a function of detachment level and the underlying portfolio's cumulative
expected default probability $p$ rather than the term to maturity $T$. Given
our assumption of the static term structure of the hazard rates $h$ these two
formulations are equivalent. In this section we prefer to emphasize the
dependence on maturity horizon in order to facilitate the comparison with base
correlation models and also to analyze the dependence on the level of hazard
rates separately from the term to maturity dimension.

Of course, the correlation surface of a static Gaussian copula model
\cite{Li-2000} is a flat surface with constant correlation across both
detachment level $K$ and term to maturity $T$. Any deviation from a flat
surface is therefore an indication of a non-trivial loss generating model, and
we can judge which features of the model are the important ones by examining
how strong a deviation from flatness they lead to.

\subsection{\label{section_static_models}Models with static dependence
structure}

Let us begin with the analysis of one of the popular static loss generation
models. On figure \ref{FigureStaticCorrSlices} we show the correlation surface
computed for the Student-t copula with linear correlation $\rho=0.3$ and
$\nu=12$ degrees of freedom. Student-t copula is in the same elliptic family
as the Gaussian copula but has non-zero tail dependence governed by the
degrees of freedom parameter. As a model of single-period asset returns the
Student-t distribution has been shown to provide a significantly better fit to
observations than the standard normal \cite{Mashal-Naldi-Zeevi-2003}.%
\begin{figure}
[t]
\begin{center}
\includegraphics
[natheight=247.562500pt,natwidth=556.937500pt,height=157.4375pt,width=352.6875pt]%
{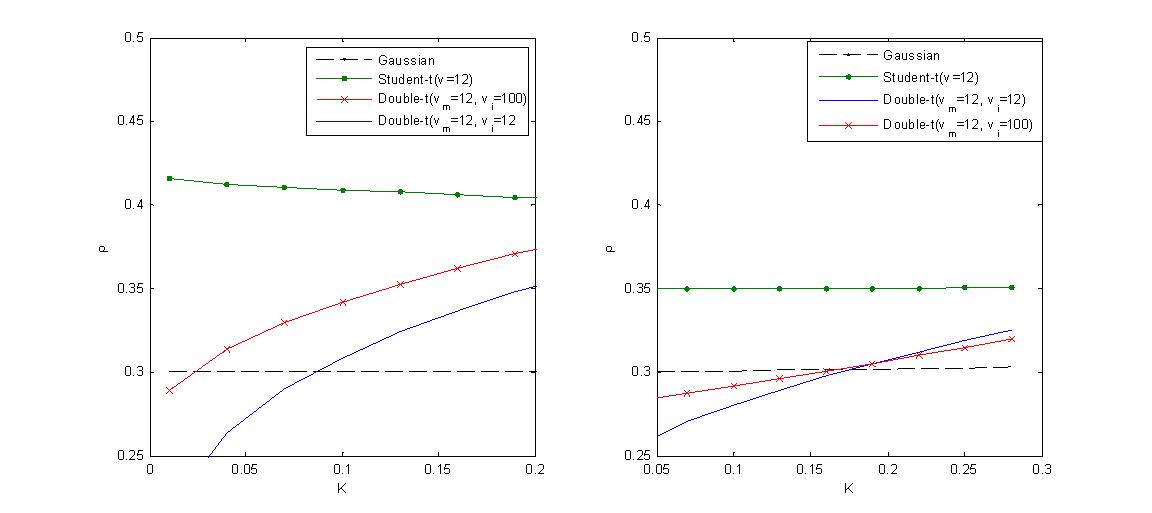}%
\caption{{\small Correlation surface slices corresponding to 1-year (left
chart) and 5-year (right chart) horizons with default probabilities 0.02 and
}$1-\left(  1-0.02\right)  ^{5}=0.0961$ {\small for Gaussian, Student-t Copula
with 12 degrees of freedom, and Double-t Copulas with (v}$_{m}=${\small 12,
v}$_{i}=${\small 100) and (v}$_{m}=${\small 12, v}$_{i}=${\small 12) degrees
of freedom. Linear correlation is 0.3 for all copulas.}}%
\label{FigureStaticCorrSlices}%
\end{center}
\end{figure}

However, from the figure \ref{FigureStaticCorrSlices} we can see that the
static Student-t copula does not generate a notable skew in the direction of
detachment level $K$, and in fact generates a mild downward sloping skew for
very short terms, which is contrary to what is observed in the market. The
main reason for this is the rigid structure of this model, with the tails of
the idiosyncratic returns tied closely to the tails of the market factor. This
follows from the representation of the Student-t copula as a mixture model
\cite{OKane-Schloegl-2003}. Instead of producing a varying degree of
correlation depending on the default threshold, the Student-t copula model
simply produces a higher overall level of correlation.

On the other hand, the more flexible double-t copula model
\cite{HullWhite2004} produces a steep upward sloping skew, as can be seen from
figure \ref{FigureStaticCorrSlices}. The main feature of the double-t copula
that is responsible for the skew is the cleaner separation between the common
factor and idiosyncratic returns -- there is no longer a single mixing
variable which ties the two sources of risk together. As a result, the
idiosyncratic returns get efficiently diversified in the LHP framework and
their contribution becomes progressively smaller for farther downside returns.
Since the higher value of the detachment level $K$ corresponds to farther
downside tails, the greater dominance of the market factor translates into
higher effective correlation for higher $K$, i.e. upward sloping correlation skew.

Furthermore, by making the fully independent idiosyncratic returns more fat
tailed one achieves a steeper skew -- compare the two examples of the double-t
copula, with the degrees of freedom of the idiosyncratic returns set to 100
(i.e. nearly Gaussian case) and to 12 (i.e. strongly fat-tailed case),
respectively. Indeed, for the same detachment level $K$ the idiosyncratic
returns with lower degrees of freedom (stronger fat tails) are less dominated
by the market factor, resulting in relatively lower effective correlation.
Since the difference between two cases diminishes as $K$ grows, this
translates into steeper correlation skew for fatter-tailed idiosyncratic returns.

Finally, we observe that the slope of the correlation skew gets flatter as the
time horizon grows. Within the context of double-t copula this is simply
because the same detachment level $K$ corresponds to less extreme tails when
the term to maturity is greater. Following the same logic as above, this means
less steep correlation skew.

All of these features will have their close counterparts in the dynamic models
which we will consider next.

\subsection{\label{section_dynamic_models}Multi-period (dynamic) loss
generating models}

Let us now turn to loss generating models based on latent variables with
multi-period dynamics. We have concluded in the previous section that a clean
separation of the market factor and the idiosyncratic returns appears to be a
pre-requisite for producing an upward sloping correlation skew. Fortunately,
the dynamic multi-variate models which we considered in section
\ref{section_equity_return} all have this property, both for single-period and
for aggregated returns.

On figure \ref{CorrSpectrum_GARCH_Gaussian} we show the correlation surface
computed for a loss generating model based on GARCH dynamics with Gaussian
residuals, with a linear correlation set to $\rho=0.3$, and GARCH model
parameters taken from the weekly SP500 estimates in appendix
\ref{appendix_SP500}. The loss distributions and the correlation surfaces are
calculated using a Monte Carlo simulation with 100,000 trials.%

\begin{figure}
[t]
\begin{center}
\includegraphics
[natheight=177.875000pt,natwidth=461.687500pt,height=140.75pt,width=363.0625pt]%
{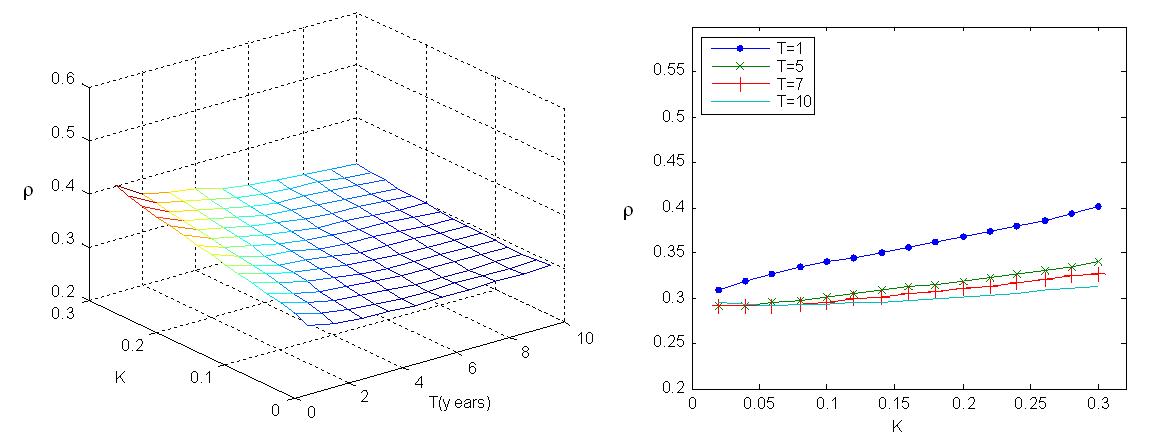}%
\caption{{\small Correlation surface for GARCH model (}$\alpha${\small =0.045,
}$\beta${\small =0.948) with Gaussian shocks and the slices of the surface for
1, 3, 5 and 7 year maturities.}}%
\label{CorrSpectrum_GARCH_Gaussian}%
\end{center}
\end{figure}
%

\begin{figure}
[t]
\begin{center}
\includegraphics
[natheight=179.250000pt,natwidth=460.937500pt,height=142.5pt,width=363.8125pt]%
{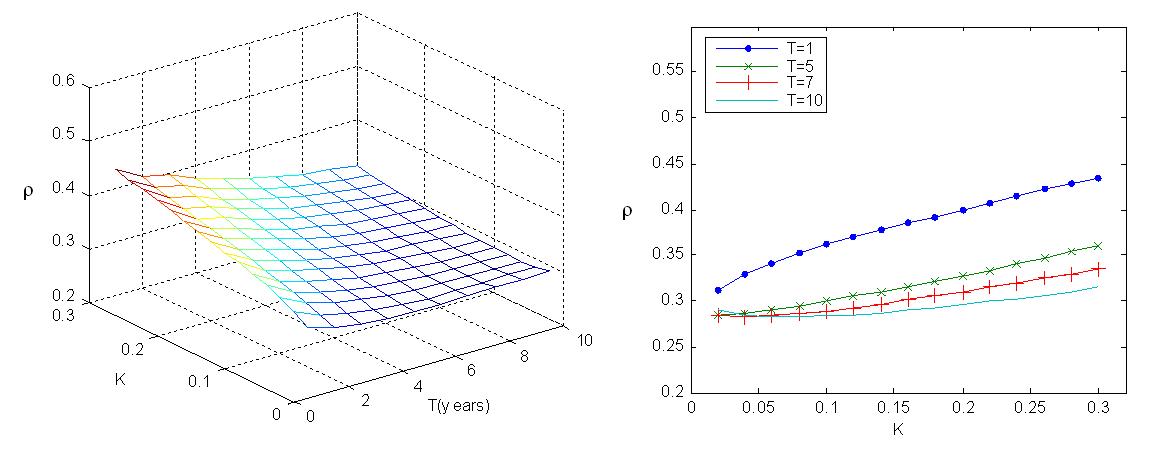}%
\caption{{\small Correlation surface for GARCH model (}$\alpha${\small =0.045,
}$\beta${\small =0.948) with Student-t shocks (v=8.3) and the slices of the
surface for 1, 3, 5 and 7 year maturities.}}%
\label{CorrSpectrum_GARCH_StudentT}%
\end{center}
\end{figure}

As we can see, this model does exhibit a visible deviation from the flat
correlation surface for short maturities. However, as we already noted in
section \ref{section_equity_return}, the distribution of aggregate returns for
the symmetric GARCH model quickly converges to normal. Therefore, it is not
surprising to see that the correlation surface also flattens out fairly
quickly and becomes virtually indistinguishable from a Gaussian copula for
maturities beyond 5 years. Thus, we conclude that the symmetric GARCH model
with Gaussian residuals is inadequate for description of liquid tranche
markets where one routinely observes steep correlation skews at maturities as
long as 7 and 10 years.

Based on the empirical results of Section \ref{SectionMultModel} we know that
a GARCH model with Student-t residuals provides a better fit to historical
time series of equity returns. A natural question is whether allowing for such
volatility dynamics can lead to a persistent correlation skew commensurate
with the levels observed in synthetic CDO markets.

The results of section \ref{SectionMultModel} suggest that the additional
kurtosis of the single-period returns represented by the Student-t residuals
does not matter very much for aggregate return distributions at sufficiently
long time horizons. Indeed, figure \ref{CorrSpectrum_GARCH_StudentT} shows
that the GARCH model with Student-t residuals exhibits a correlation skew that
is quite a bit steeper at the short maturities, yet is almost as flat and
featureless at the long maturities as its non-fat-tailed counterpart -- there
is a small amount of skew at 10 years, but it is too small compared to the
steepness observed in the liquid tranche markets. Thus, we conclude that one
has to focus on the dynamic features of the market factor process in order to
achieve the desired correlation skew effect.

Our next candidates are the TARCH models with either Gaussian or Student-t
return innovations. We have seen in section \ref{section_equity_return} that
the asymmetric volatility dynamics of these models leads to a much more
persistent skewness and kurtosis of aggregated equity returns that actually
grow rather than decay at very short horizons, and survive for as long as 10
years for the range of parameters corresponding to the post-1990 sample of
SP500 weekly log-returns. Hence, our hypothesis is that a latent variable
model with TARCH market dynamics might be capable of producing a non-trivial
credit correlation skew for maturities of up to 10 years.

Figures \ref{CorrSpectrum_TARCH_Gaussian} and
\ref{CorrSpectrum_TARCH_StudentT} show the correlation surfaces for the
TARCH-based loss generating models. The most immediate observation is that
both versions of the model produce a rather persistent correlation skew.
Although the correlation surface flattens out with growing term to maturity,
the steepness of the skew is still quite significant even at 10 years. Just as
in the case of the symmetric GARCH model, the fat-tailed residuals lead only
to marginal steepening of the correlation surface compared to the case with
Gaussian residuals.%
\begin{figure}
[t]
\begin{center}
\includegraphics
[natheight=176.437500pt,natwidth=457.375000pt,height=140.4375pt,width=361.75pt]%
{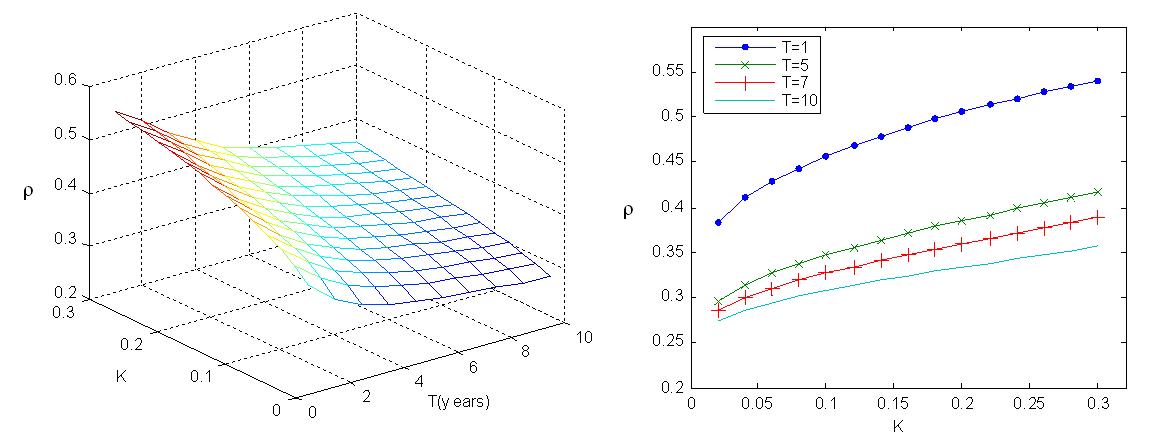}%
\caption{{\small Correlation surface for TARCH model (}$\alpha${\small =0.004,
}$\alpha_{d}${\small =0.094, }$\beta${\small =0.927) with Gaussian shocks and
the slices of the surface for 1, 3, 5 and 7 year maturities.}}%
\label{CorrSpectrum_TARCH_Gaussian}%
\end{center}
\end{figure}
%

\begin{figure}
[t]
\begin{center}
\includegraphics
[natheight=179.250000pt,natwidth=461.687500pt,height=140.4375pt,width=359.25pt]%
{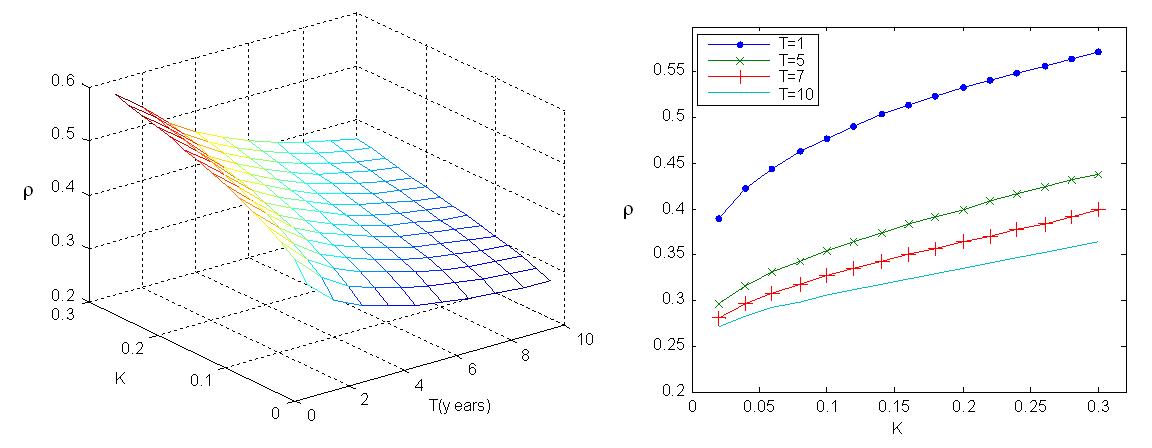}%
\caption{{\small Correlation surface for TARCH model (}$\alpha${\small =0.004,
}$\alpha_{d}${\small =0.094, }$\beta${\small =0.927) with Student-t
shocks(v=8.3) and the slices of the surface for 1, 3, 5 and 7 year
maturities.}}%
\label{CorrSpectrum_TARCH_StudentT}%
\end{center}
\end{figure}

Contrast these properties of the dynamic GARCH-based models with the features
of the static double-t copula. Upon a closer inspection of figures
\ref{FigureStaticCorrSlices} and \ref{CorrSpectrum_TARCH_Gaussian} we can see
that the TARCH model with Gaussian shocks and Gaussian idiosyncrasies produces
a slightly steeper 5-year correlation skew than the double-t copula, even when
the latter is taken with fat-tailed idiosyncrasies. When we turn on the
Student-t return residuals for the market factor dynamics (see figure
\ref{CorrSpectrum_TARCH_StudentT}) the differences in the 5-year skew become
quite significant.

The explanation of the correlation skew in the dynamic TARCH-based models is
similar to the static double-t copula when one considers a particular time
horizon. The separation of aggregate returns for the common market factor and
idiosyncratic factors remains valid for all time horizons. The diminishing
importance of the idiosyncratic returns compared to the market factor for the
greater values of detachment level $K$ explains most of the steepness of the
correlation skew. Practitioners using the static models often have to assume
heuristic term structure dependence for base correlations, typically without
fundamental reasons why one choice or another is preferred, and consequently
leading to biased relative value assessments between tranches of different
maturities. Our TARCH-based model, in contrast, produces a characteristic
pattern of correlation skew dependence on $T$ which is driven by the speed of
convergence of the idiosyncratic factors to normal and can serve as a starting
point for such relative value assessments. We confine our discussion of this
subject to highlighting of this possibility, since its detailed analysis in
the context of actual market prices is well beyond the scope of our paper.

\subsection{\label{section_sensitivity}Sensitivity to model parameters and
hedging applications}

The predictable shape of the correlation surface, demonstrated above, is
complemented by an equally important feature of the dynamic loss generating
models -- they lead to a well-defined dependence of the correlation skew on
model parameters. The ability to calculate the sensitivity of expected losses
to underlying parameters is crucial in risk management applications. In
particular, sensitivity of the correlation surface with respect to the
underlying portfolio hazard rate leads to a non-trivial adjustment of the CDO
tranche deltas.

First, let us demonstrate this sensitivity for TARCH model with Gaussian
idiosyncrasies. Figure \ref{CorrSpect_Hazard_TARCH_Gaussian} shows the
correlation skew $\rho\left(  K,p_{t}\left(  h\right)  ,\overset{\_}%
{R}\right)  $ of the 5-year tranches with various detachment levels $K$ as a
function of varying portfolio hazard rate $h$. The range of variation is
chosen from 100bp to 500bp, which corresponds roughly to portfolio spreads
ranging from 40bp to 200bp, with $\overset{\_}{R}=0.4$. From the visual
comparison of figures \ref{CorrSpect_Hazard_TARCH_Gaussian} and
\ref{CorrSpectrum_TARCH_Gaussian} it appears that the dependence of the
correlation skew for a fixed term to maturity but varying level of hazard
rates is very similar to the dependence of the correlation surface on the term
to maturity. This similarity is natural, since the first order effect is the
dependence on the level of the cumulative default probability $p_{t}\left(
h\right)  =1-e^{-ht}$ which depends on the product of $h\cdot t$ rather than
on the hazard rate or the term to maturity separately. For each level of this
product, we get a specific level of the default threshold in the structural
credit risk model. The higher this threshold, the closer is the sampled region
to the center of the latent variables distribution and the less it is affected
by the tail risk -- thus leading to a lower level and flatter skew of the
credit correlation.

However, there is a second order effect which makes these two dependencies
somewhat different. Let us recall first that idiosyncrasies with less fat
tails (higher degrees of freedom) correspond to flatter correlation skew, as
we argued in the previous section. Since the idiosyncrasies converge to normal
distribution faster than the market factor as the return aggregation horizon
grows, we can deduce that the dependence on the term to maturity with fixed
hazard rate should exhibit a faster flattening of the correlation surface than
the dependence on the hazard rate with fixed term to maturity.

The right hand side figure in \ref{CorrSpect_Hazard_TARCH_Gaussian} shows a
comparison of the change in correlation when going from 5-year horizon to
10-year horizon with constant hazard rate set at 100bp, against the change in
correlation of fixed 5-year slice when hazard rate goes from 100bp to 200bp.
One can observe that the term-to-maturity extension indeed causes a greater
degree of flattening than the hazard rate increase.%

\begin{figure}
[t]
\begin{center}
\includegraphics
[natheight=173.562500pt,natwidth=460.937500pt,height=134.5625pt,width=354.75pt]%
{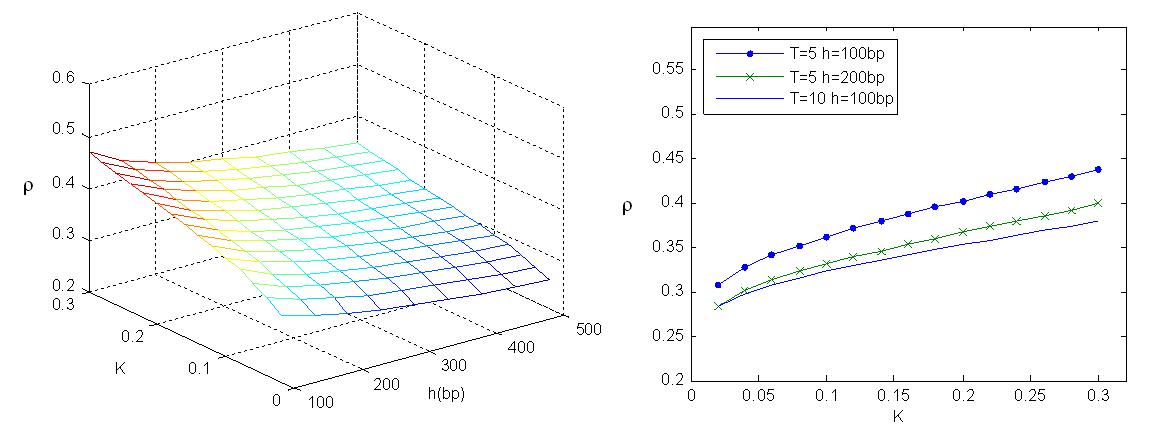}%
\caption{{\small Left figure: the dependence of the 5-year correlation skew on
the level of the hazard rates. Right figure shows 3 correlation surface
slices, contrasting the flattening of the skew with growing hazard rates and
term to maturity.}}%
\label{CorrSpect_Hazard_TARCH_Gaussian}%
\end{center}
\end{figure}

The dependence of the correlation surface on hazard rate has a strong effect
on CDO tranche deltas. The precise calculation of tranche deltas in our
framework requires one to derive the tranche loss probabilities from the shape
of the correlation surface, as outlined in section
\ref{section_correlation_spectrum}, and then use the standard pricing
techniques described in Schonbucher \cite{SchonbucherBook}. However, the
magnitude of the adjustment can be estimated more easily by neglecting the
interest rates and looking only on the protection leg of the given equity
tranche. The tranche delta is dominated by the hazard rate sensitivity of the
present value of its protection leg, which in turn is proportional to the
expected tranche loss:%

\begin{equation}
\Delta_{\left(  0,K\right]  }\left(  p_{t}\right)  \propto\frac{dEL_{\left(
0,K\right]  }(t)}{dh}=E_{h}^{G}L_{\left(  0,K\right]  }(t)+\rho_{h}\left(
K,p_{t},\overset{\_}{R}\right)  E_{\rho}^{G}L_{\left(  0,K\right]  }(t)
\label{delta_proportional}%
\end{equation}
where, in accordance with the definition of the correlation surface, we have:%

\begin{equation}
EL_{\left(  0,K\right]  }\left(  t\right)  =E^{G}L_{\left(  0,K\right]
}\left(  \rho\left(  K,p_{t}\left(  h\right)  ,\overset{\_}{R}\right)
,p_{t}\left(  h\right)  \right)  \label{ELKt}%
\end{equation}

Let us define the delta adjustment factor as the percentage adjustment to the
tranche delta compared to fixed-correlation Gaussian delta of the tranche:%

\begin{equation}
\Delta_{\left(  0,K\right]  }\left(  p_{t}\right)  =\Delta_{\left(
0,K\right]  }^{G}\left(  p_{t}\right)  \left(  1+\delta_{adj}\left(
K,p_{t}\right)  \right)  \label{delta_decomp}%
\end{equation}

Since the first term in eq. (\ref{delta_proportional}) corresponds to the
Gaussian delta of the tranche, we can see that the delta adjustment factor is
equal to the correlation surface sensitivity times the tranche loss
sensitivity ratio:%

\begin{equation}
\delta_{adj}\left(  K,p_{t}\right)  =\rho_{h}\left(  K,p_{t},\overset{\_}%
{R}\right)  \frac{E_{\rho}^{G}L_{\left(  0,K\right]  }(t)}{E_{h}^{G}L_{\left(
0,K\right]  }(t)}=\rho_{h}\left(  K,p_{t},\overset{\_}{R}\right)  \Psi\left(
K,p_{t}\right)  \label{delta_adj}%
\end{equation}
where%

\begin{gather}
E_{\rho}^{G}L_{\left(  0,K\right]  }(t)=-\frac{1-\overset{\_}{R}}{2\sqrt{\rho
}}\phi\left(  \Phi^{-1}\left(  p_{t}\right)  ,-d_{1};-\sqrt{\rho}\right) \\
E_{h}^{G}L_{\left(  0,K\right]  }(t)\equiv\frac{dp_{t}}{dh}E_{p}^{G}L_{\left(
0,K\right]  }(t)=\left(  1-p_{t}\right)  t\left(  1-\overset{\_}{R}\right)
\Phi\left(  \frac{-d_{1}+\sqrt{\rho}\Phi^{-1}\left(  p_{t}\right)  }%
{\sqrt{1-\rho}}\right) \\
d_{1}=\frac{1}{\sqrt{\rho}}\Phi^{-1}\left(  p_{t}\right)  -\frac{\sqrt{1-\rho
}}{\sqrt{\rho}}\Phi^{-1}\left(  \frac{K}{1-\overset{\_}{R}}\right)
\end{gather}
and the functions are evaluated at $\rho=\rho\left(  K,p_{t}\left(  h\right)
,\overset{\_}{R}\right)  $. Thus, the tranche loss sensitivity ratio is given by:%

\begin{equation}
\Psi\left(  K,p_{t}\right)  =\frac{E_{\rho}^{G}L_{\left(  0,K\right]  }%
(t)}{E_{h}^{G}L_{\left(  0,K\right]  }(t)}=-\frac{1}{2\left(  1-p_{t}\right)
t\sqrt{\rho}}\frac{\phi\left(  \Phi^{-1}\left(  p_{t}\right)  ,-d_{1}%
;-\sqrt{\rho}\right)  }{\Phi\left(  \frac{-d_{1}+\sqrt{\rho}\Phi^{-1}\left(
p_{t}\right)  }{\sqrt{1-\rho}}\right)  } \label{exploss_adj}%
\end{equation}

Figure \ref{Fig_AdjDelta} illustrates these calculations for the case of
TARCH-based dynamic loss generating model. We can see that the systematic
delta estimates in our model are anywhere from 40\% to 100\% greater than the
conventional fixed-correlation delta estimates stemming from the Gaussian
copula model. The adjustment is greatest for the lowest values of $K$, and
drops quickly as the detachment level grows. While comparison with the actual
tranche market price sensitivity is beyond the scope of our paper, we would
like to note that both the the significant under-estimation of the equity
tranche deltas and the relatively smaller amount of the delta error for more
senior tranches are in line with the market experience during the correlation
dislocation in May-September of 2005.%

\begin{figure}
[t]
\begin{center}
\includegraphics
[natheight=2.267500in,natwidth=6.000100in,height=1.8291in,width=4.8032in]%
{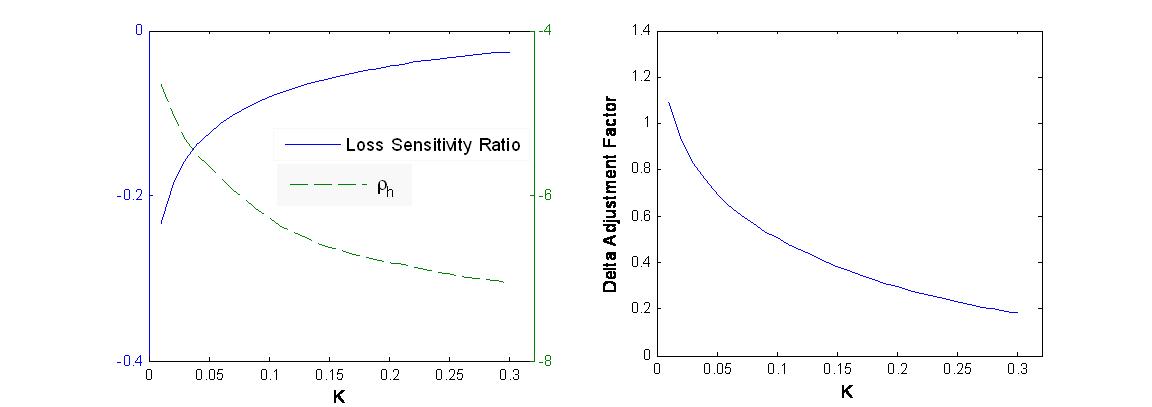}%
\caption{{\small Left figure: the dependence of the tranche loss sensitivity
ratio and }$\rho_{h}$ {\small on }$K${\small . Right figure shows the product
of the delta adjustment factor as a function of }$K$. {\small The base case
corresponds to TARCH model (}$\alpha${\small =0.004, }$\alpha_{d}%
${\small =0.094, }$\beta$={\small 0.927}) {\small with Gaussian
idiosyncracies, aggregated over 5 year horizon. Portfolio hazard rate is equal
to 100bps.}}%
\label{Fig_AdjDelta}%
\end{center}
\end{figure}

\section{\label{section_conclusions}Summary and Conclusions}

In this paper we have introduced and studied a new class of credit correlation
models defined as an extension of the structural credit model where the latent
variables follow a factor-ARCH process with asymmetric volatility dynamics. To
build the foundation for our model, we have studied the time aggregation
properties of the multivariate dynamic models of equity returns. We showed
that the dynamics of equity return volatilities and correlations leads to
significant departures from the Gaussian distribution even for horizons
measured in several years. The asymmetry appears to "survive aggregation"
longer than fat tails based on the parameters estimated from the real data.
The main source of skewness and kurtosis of the return distribution for long
horizons is the dynamic asymmetry of volatility response to return shocks.

We introduced the notion of the correlation surface as a tool for comparing
loss generating models, whether defined via a single-period (static) copula,
or via multi-period (dynamic) latent-variable framework, and for simple and
consistent approach to non-parametric pricing of CDO tranches. We showed that
portfolio loss distributions with smooth pdf can be easily reconstructed from
the correlation surface using its level and slope along the K-dimention.

We considered the differences in the correlation surfaces generated by static
models, including Gaussian, Student-t and Double-t copula, and dynamic models
including GARCH(1,1) with Gaussian \ and Student-t shocks, and TARCH(1,1) with
Gaussian and Student-t shocks. From this comparison we can conclude that the
most relevant stylized facts for explanation of the market observable
correlation skew are, in order of their importance:

\begin{itemize}
\item the independence of the market factor and idiosyncratic returns (no
common mixing variables);

\item the persistent asymmetry of aggregate return distribution of the market
factor, which in case of TARCH models occurs as a consequence of the
asymmetric volatility dynamics;

\item the fat tails in the market factor returns;

\item the fat tails in the idiosyncratic returns;

\item the slower convergence to Normal distribution of the market factor
compared to idiosyncratic returns;
\end{itemize}

Importantly, in our dynamic framework, the correlation surface is not only
explained, but predicted -- based on empirical parameters of the TARCH process
and the parameters describing the reference credit portfolio. The model also
predicts a specific sensitivity of the correlation surface to changes in
various parameters, including the average hazard rate of the underlying portfolio.

The inability of static models to incorporate changing base correlations are
at the heart of the difficulties faced by these models during the credit
market dislocations. In particular, our model reveals that the systematic
deltas of equity tranches are understated by the industry standard static
copula models, since the growing portfolio spread (hazard rate) should lead to
an additional drop in equity tranche prices due to decreasing implied
correlation level. Similarly, the static models require making additional
assumptions about the correlation skew at different maturities. The early
market convention of keeping this skew constant which some practitioners still
adhere to is very far from realistic as can be seen from the results of
section \ref{section_dynamic_models}. The more reasonable assumption is that
the correlation level decreases and the skew flattens for longer horizons.

We should note however, that these conclusions are based on an implicit
assumption that the model includes a single common return factor, and that the
parameters of the dynamic model are constant over time. While this is a weaker
assumption that an outright imposition of the constant correlation skew, it
may still be too strict in some circumstances. A possible direction for
generalization of our model is to move from a single market factor to a
multi-factor framework which can make the model much more flexible in terms of
both the detachment level $K$ and term structure $T$ dependence of the
correlation surface $\rho=\rho(K,T)$. The well-documented importance of both
macro and industry factors for explanation of equity returns suggests that
such a generalization is not only desirable from calibration point of view but
also warranted empirically. While the analytical tractability of the model
will suffer, the numerical accuracy will likely remain intact when using Monte
Carlo simulations.

Whether in a single factor or a multi-factor setting, many of our conclusions
reflect the limitations of the large homogeneous portfolio approximation which
we have adopted in this paper. In particular, it is clear that even
deterministic but heterogenous idiosyncrasies, market factor loadings and
hazard rates could lead to significant changes in portfolio loss distribution
and consequently to the correlation surface of the model. An extension of our
model to such heterogenous case is possible, although the computational
efforts will increase very significantly.

In conclusion we note that the ARCH family of time series models
\cite{EngleArch} had proven quite successful in explaining the behavior of
implied volatility smile and skew and stock index option pricing
\cite{EngleSkew}. Given the similar empirical motivation of our model, and
multi-faceted analogies with equity derivatives pricing outlined throughout
the paper, we believe that our approach can lead to similar advances in the
portfolio credit risk modeling, and shed new light on pricing of CDO tranches.

\appendix

\section{\label{appendix_aggregate_TARCH}Kurtosis and Skewness of Aggregated
TARCH Returns}

In this notes we analyze kurtosis and skewness of aggregated returns $R_{T}=%
{\displaystyle\sum\limits_{t=1}^{T}}
r_{t}$ when $r_{t}$ is assumed to follow TARCH(1,1) process%

\begin{align*}
r_{t}  &  =\sigma_{t}\varepsilon_{t}\\
\sigma_{t}^{2}  &  =\left(  1-\zeta\right)  \sigma^{2}+\alpha r_{t-1}%
^{2}+\alpha_{d}r_{t-1}^{2}1_{\left\{  r_{t-1}\leq0\right\}  }+\beta
\sigma_{t-1}^{2}%
\end{align*}
where returns innovations $\varepsilon_{t}$ are assumed to be iid, have zero
mean and unit variance. We are interested in variance, skewness and kurtosis
of time aggregated returns. To make sure that those moments are finite we need
corresponding moments of the return innovations to be finite. Particularly, we
assume that $\varepsilon_{t}$ has finite kurtosis. Let us introduce the
following notations for the central and truncated moments of $\varepsilon_{t}$%
\begin{align*}
m_{\varepsilon}  &  \equiv E\left(  \varepsilon_{t}\right)  =0\\
v_{\varepsilon}  &  \equiv E\left(  \varepsilon_{t}^{2}\right)  =1\\
v_{\varepsilon}^{d}  &  \equiv E\left(  \varepsilon_{t}^{2}1_{\left\{
\varepsilon_{t}\leq0\right\}  }\right) \\
s_{\varepsilon}  &  \equiv E\left(  \varepsilon_{t}^{3}\right) \\
s_{\varepsilon}^{d}  &  \equiv E\left(  \varepsilon_{t}^{3}1_{\left\{
\varepsilon_{t}\leq0\right\}  }\right) \\
k_{\varepsilon}  &  \equiv E\left(  \varepsilon_{t}^{4}\right) \\
k_{\varepsilon}^{d}  &  \equiv E\left(  \varepsilon_{t}^{4}1_{\left\{
\varepsilon_{t}\leq0\right\}  }\right)
\end{align*}

\begin{lemma}
The following recursions hold for TARCH(1,1) model%
\begin{align*}
cov_{t-1}\left(  r_{t}^{k},r_{t+u}^{2}\right)   &  =\rho cov_{t-1}\left(
r_{t}^{k},r_{t+u-1}^{2}\right)  \text{ for u%
$>$%
1}\\
cov_{t-1}\left(  r_{t}^{k}r_{t+1}^{2}\right)   &  =\alpha var_{t-1}\left(
r_{t}^{k+2}\right)  +\alpha_{d}var_{t-1}\left(  r_{t}^{k+2}1_{\left\{
r_{t}\leq0\right\}  }\right)
\end{align*}

\begin{proof}%
\begin{gather*}
cov_{t-1}\left(  r_{t}^{k},r_{t+u}^{2}\right)  =cov_{t-1}\left(  r_{t}%
^{k}\left[  \left(  1-\zeta\right)  \sigma^{2}+\alpha r_{t+u-1}^{2}+\alpha
_{d}r_{t+u-1}^{2}1_{\left\{  r_{t+u-1}\leq0\right\}  }+\beta\sigma_{t+u-1}%
^{2}\right]  \right) \\
=0+\alpha cov_{t-1}\left(  r_{t}^{k},r_{t+u-1}^{2}\right)  +\alpha
_{d}cov_{t-1}\left(  r_{t}^{k},r_{t+u-1}^{2}1_{\left\{  r_{t+u-1}%
\leq0\right\}  }\right)  +\beta cov_{t-1}\left(  r_{t}^{k},\sigma_{t+u-1}%
^{2}\right)
\end{gather*}
if u%
$>$%
1 then%
\begin{align*}
cov_{t-1}\left(  r_{t}^{k},r_{t+u-1}^{2}1_{\left\{  r_{t+u-1}\leq0\right\}
}\right)   &  =v_{\varepsilon}^{d}cov_{t-1}\left(  r_{t}^{k},r_{t+u-1}%
^{2}\right) \\
cov_{t-1}\left(  r_{t}^{k},\sigma_{t+u-1}^{2}\right)   &  =cov_{t-1}\left(
r_{t}^{k},r_{t+u-1}^{2}\right)
\end{align*}
If u=1 then
\[
cov_{t-1}\left(  r_{t}^{k},\sigma_{t+u-1}^{2}\right)  =0
\]

\end{proof}
\end{lemma}

\begin{proposition}
Suppose $0\leq\zeta<1$ and the return innovations have finite skewness,
$s_{\varepsilon},$ and finite "truncated" third moment, $s_{\varepsilon}^{d},$
then conditional third moment of T-period aggregate return $R_{t,t+T}$ has the
following representation for TARCH(1,1)
\[
E_{t}R_{t,t+T}^{3}=s_{\varepsilon}%
{\displaystyle\sum\limits_{u=1}^{T}}
E_{t}\left(  \sigma_{t+u}^{3}\right)  +3\left(  \alpha s_{\varepsilon}%
+\alpha_{D}s_{\varepsilon}^{d}\right)
{\displaystyle\sum\limits_{u=1}^{T}}
\frac{1-\zeta^{T-u}}{1-\zeta}E_{t}\left(  \sigma_{t+u}^{3}\right)
\]
In addition if $E\sigma_{t}^{3}$ is finite then unconditional skewness of
$R_{t,t+T}$ is given by
\[
S_{T}\equiv\frac{ER_{t,t+T}^{3}}{E(R_{t,t+T}^{2})^{3/2}}=\left[  \frac
{1}{T^{1/2}}s_{\varepsilon}+3\frac{1}{T^{3/2}}\left(  \alpha s_{\varepsilon
}+\alpha_{d}s_{\varepsilon}^{d}\right)  \frac{T(1-\zeta)-1+\rho^{T}}%
{(1-\zeta)^{2}}\right]  E\left(  \frac{\sigma_{t}}{\sigma}\right)  ^{3}%
\]

\begin{proof}
Using Lemma 7 we have
\begin{align*}
E_{t}\left(
{\displaystyle\sum\limits_{u=t+1}^{t+T}}
r_{u}\right)  ^{3}  &  =E_{t}\left(
{\displaystyle\sum\limits_{t+1\leq t_{1}\leq t_{2}\leq t_{3}\leq t+T}}
r_{t_{1}}r_{t_{2}}r_{t_{3}}\right) \\
&  =%
{\displaystyle\sum\limits_{u=1}^{T}}
E_{t}r_{t+u}^{3}+%
{\displaystyle\sum\limits_{t+1\leq t_{1}<t_{2}\leq t+T}}
3E_{t}\left(  r_{t_{1}}r_{t_{2}}^{2}\right) \\
&  =%
{\displaystyle\sum\limits_{u=1}^{T}}
E_{t}\left(  r_{t+u}^{3}\right)  +3%
{\displaystyle\sum\limits_{t+1\leq t_{1}<t_{2}\leq t+T}}
\zeta^{t_{2}-t_{1}-1}\left(  \alpha E_{t}\left(  r_{t_{1}}^{3}\right)
+\alpha_{d}E_{t}\left(  r_{t_{1}}^{3}1_{\left\{  r_{t_{1}}\leq0\right\}
}\right)  \right) \\
&  =%
{\displaystyle\sum\limits_{u=1}^{T}}
E_{t}\left(  r_{t+u}^{3}\right)  +3%
{\displaystyle\sum\limits_{u=1}^{T}}
\frac{1-\zeta^{T-u}}{1-\zeta}\left(  \alpha E_{t}\left(  r_{t+u}^{3}\right)
+\alpha_{d}E_{t}\left(  r_{t+u}^{3}1_{\left\{  r_{t+u}\leq0\right\}  }\right)
\right) \\
&  =s_{\varepsilon}%
{\displaystyle\sum\limits_{u=1}^{T}}
E_{t}\left(  \sigma_{t+u}^{3}\right)  +\left(  \alpha s_{\varepsilon}%
+\alpha_{d}s_{\varepsilon}^{d}\right)
{\displaystyle\sum\limits_{u=1}^{T}}
\frac{1-\zeta^{T-u}}{1-\zeta}E_{t}\left(  \sigma_{t+u}^{3}\right)
\end{align*}
Using the law of iterated expectations%
\[
E\left(
{\displaystyle\sum\limits_{u=t+1}^{t+T}}
r_{u}\right)  ^{3}=E\left(  E_{t}\left(
{\displaystyle\sum\limits_{u=t+1}^{t+T}}
r_{u}\right)  ^{3}\right)  =\left[  Ts_{\varepsilon}+3\left(  \alpha
s_{\varepsilon}+\alpha_{d}s_{\varepsilon}^{d}\right)  \frac{T(1-\zeta
)-1+\zeta^{T}}{(1-\zeta)^{2}}\right]  E\left(  \sigma_{t}\right)  ^{3}%
\]
$S_{T}$ is then computed using the simple formula for the unconditional
variance $E(R_{t,t+T}^{2})=\sigma^{2}$.
\end{proof}
\end{proposition}

To derive unconditional kurtosis we define the following unconditional autocorrelations%

\begin{align*}
\gamma_{n}  &  =\gamma_{-n}=corr(r_{t-n}^{2},r_{t}^{2})\\
\varphi_{n}  &  =corr(r_{t-n},r_{t}^{2})\text{ for }n\geq1\\
\psi_{i,j}  &  \equiv E\left(  r_{t-i}r_{t-j}r_{t}^{2}\right)  \text{ for
}1\leq j<i
\end{align*}

\begin{lemma}
$\gamma_{n},$ $\varphi_{n}$ and $\psi_{i,j}$ decay exponentially as $n$ and
$i-j$ increase%
\begin{align*}
\gamma_{n}  &  =\zeta\gamma_{n-1}=\zeta^{n-1}\gamma_{1}\text{ for }n\geq1\\
\varphi_{n}  &  =\zeta\varphi_{n-1}=\zeta^{n-1}\varphi_{1}\text{ for }n\geq1\\
\psi_{i,j}  &  =\zeta\psi_{i-1,j-1}=\zeta^{j-1}\psi_{i-j+1,1}\text{ for }1\leq
j<i
\end{align*}
where $\gamma_{1},$ $\varphi_{1}$ and $\psi_{k,1}$ are given by%
\begin{align*}
\gamma_{1}  &  =\alpha\left(  k_{r}-1\right)  +\alpha_{d}\left(  k_{r}%
^{d}-v_{r}^{d}\right)  +\beta k_{r}/k_{\varepsilon}\\
\varphi_{1}  &  =\alpha s_{r}+\alpha_{d}s_{r}^{d}\\
\psi_{k,1}  &  =\alpha E\left(  r_{t-k+1}r_{t}^{3}\right)  +\alpha_{d}E\left(
r_{t-k+1}r_{t}^{3}1_{\left\{  r_{t}\leq0\right\}  }\right)
\end{align*}
with $v_{\varepsilon}^{d}=\frac{E\left(  r_{t}^{2}1_{\left\{  r_{t}%
\leq0\right\}  }\right)  }{Er_{t}^{2}},$ $s_{r}=\frac{E\left(  r_{t}%
^{3}\right)  }{\left(  Er_{t}^{2}\right)  ^{3/2}}$, $s_{r}^{d}=\frac{E\left(
r_{t}^{3}1_{\left\{  r_{T}<0\right\}  }\right)  }{\left(  Er_{t}^{2}\right)
^{3/2}},$ $k_{r}=\frac{E\left(  r_{t}^{4}\right)  }{\left(  Er_{t}^{2}\right)
^{2}}$ and $k_{r}^{d}=\frac{E\left(  r_{t}^{4}1_{\left\{  r_{t}\leq0\right\}
}\right)  }{\left(  Er_{t}^{2}\right)  ^{2}}.$
\end{lemma}

\begin{proposition}
If%
\begin{align*}
\zeta &  \equiv E\left(  \beta+\alpha\varepsilon_{t}^{2}+\alpha_{D}%
\varepsilon_{t}^{2}1_{\left\{  \varepsilon_{t}\leq0\right\}  }\right)
=\beta+\alpha+\alpha_{D}v_{\varepsilon}^{d}<1\\
\xi &  \equiv E\left(  \beta+\alpha\varepsilon_{t}^{2}+\alpha_{D}%
\varepsilon_{t}^{2}1_{\left\{  \varepsilon_{t}\leq0\right\}  }\right)
^{2}=\beta^{2}+\alpha^{2}k_{\varepsilon}+\alpha_{D}^{2}k_{\varepsilon}%
^{d}+2\alpha\beta+2\alpha_{D}\beta v_{\varepsilon}^{d}+2\alpha\alpha
_{D}k_{\varepsilon}^{d}<1
\end{align*}
then unconditional kurtosis of $r_{t}$, $K_{1},$ is finite and%
\[
K_{1}\equiv\frac{Er_{t}^{4}}{\left(  Er_{t}^{2}\right)  ^{2}}=k_{\varepsilon
}\frac{1-\zeta^{2}}{1-\xi}%
\]

\begin{proof}
If the 4th moment of $r_{t}$ exists then the following equation must hold%
\begin{align*}
Er_{t}^{4}  &  =E\left(  \varepsilon_{t}^{4}\right)  E\left(  \sigma_{t}%
^{4}\right) \\
&  =k_{\varepsilon}E\left(  \left(  1-\zeta\right)  \sigma^{2}+\alpha
r_{t-1}^{2}+\alpha_{d}r_{t-1}^{2}1_{\left\{  r_{t-1}\leq0\right\}  }%
+\beta\sigma_{t-1}^{2}\right)  ^{2}\\
&  =k_{\varepsilon}(\left(  1-\zeta\right)  ^{2}\sigma^{4}+2\left(
1-\zeta\right)  \sigma^{2}E\left(  \alpha r_{t-1}^{2}+\alpha_{d}r_{t-1}%
^{2}1_{\left\{  r_{t-1}\leq0\right\}  }+\beta\sigma_{t-1}^{2}\right) \\
&  +\left(  \alpha r_{t-1}^{2}+\alpha_{d}r_{t-1}^{2}1_{\left\{  r_{t-1}%
\leq0\right\}  }+\beta\sigma_{t-1}^{2}\right)  ^{2})\\
&  =k_{\varepsilon}\left(  \left(  1-\zeta\right)  ^{2}\sigma^{4}+2\left(
1-\zeta\right)  \rho\sigma^{4}+\xi E\sigma_{t-1}^{4}\right)
\end{align*}
Therefore $Er_{t}^{4}$ necessarily solves
\[
Er_{t}^{4}=k_{\varepsilon}\left(  1-\zeta^{2}\right)  \sigma^{4}+\xi
Er_{t}^{4}%
\]

\end{proof}
\end{proposition}

\begin{proposition}
If the distribution of $\varepsilon_{t}$ is symmetric and $\alpha_{d}=0$ then
unconditional kurtosis of $R_{T},$ if exists, is given by the following
formula:
\begin{align}
K_{T}  &  =3+\frac{1}{T}(K_{1}-3)+6\frac{\gamma_{1}}{T^{2}}\frac
{T(1-\zeta)-1+\zeta^{T}}{(1-\zeta)^{2}}\text{for }T>1\\
K_{1}  &  =k_{\varepsilon}\frac{1-\zeta^{2}}{1-\xi}%
\end{align}
where $k_{\varepsilon}$ is unconditional kurtosis of $\varepsilon_{t}$ and
\begin{align*}
\xi &  \equiv E\left(  \beta+\alpha\varepsilon_{t}^{2}+\alpha_{d}%
\varepsilon_{t}^{2}1_{\left\{  \varepsilon_{t}\leq0\right\}  }\right)
^{2}=\beta^{2}+\alpha^{2}k_{\varepsilon}+\alpha_{d}^{2}k_{\varepsilon}%
^{d}+2\alpha\beta+2\alpha_{d}\beta v_{\varepsilon}^{d}+2\alpha\alpha
_{d}k_{\varepsilon}^{d}.\\
\gamma_{1}  &  \equiv corr\left(  r_{t-1}^{2},r_{t}^{2}\right)  =\alpha\left(
k_{r}-1\right)  +\alpha_{d}\left(  k_{r}^{d}-v_{r}^{d}\right)  +\beta
k_{r}/k_{\varepsilon}%
\end{align*}

\begin{proof}%
\begin{align*}
E\left(
{\displaystyle\sum\limits_{u=t+1}^{t+T}}
r_{u}\right)  ^{4}  &  =%
{\displaystyle\sum\limits_{u=1}^{T}}
E\left(  r_{t+u}^{4}\right)  +6%
{\displaystyle\sum\limits_{t+1\leq t_{1}<t_{2}\leq t+T}}
E\left(  r_{t_{1}}^{2}r_{t_{2}}^{2}\right) \\
&  =%
{\displaystyle\sum\limits_{u=1}^{T}}
E\left(  r_{t+u}^{4}\right)  +6%
{\displaystyle\sum\limits_{t+1\leq t_{1}<t_{2}\leq t+T}}
\left[  cov\left(  r_{t_{1}}^{2},r_{t_{2}}^{2}\right)  +E\left(  r_{t_{1}}%
^{2}\right)  E\left(  r_{t_{2}}^{2}\right)  \right] \\
&  =TE\left(  r_{t}^{4}\right)  +6\frac{T(T-1)}{2}E\left(  r_{t}^{2}\right)
^{2}+6cov\left(  r_{t-1}^{2},r_{t}^{2}\right)
{\displaystyle\sum\limits_{t+1\leq t_{1}<t_{2}\leq t+T}}
\zeta^{t_{2}-t_{1}-1}\\
&  =TE\left(  r_{t}^{4}\right)  +6\frac{T(T-1)}{2}E\left(  r_{t}^{2}\right)
^{2}+6cov\left(  r_{t-1}^{2},r_{t}^{2}\right)  \frac{T(1-\zeta)-1+\zeta^{T}%
}{(1-\zeta)^{2}}%
\end{align*}
substituting the derived 4th moment into the definition of the kurtosis
$K_{T}=E\left(
{\displaystyle\sum\limits_{u=t+1}^{t+T}}
r_{u}\right)  ^{4}/E\left(  r_{t}^{2}\right)  ^{2}$ completes the proof.
\end{proof}
\end{proposition}

\section{\label{appendix_SP500}Estimation Results for SP500}

In this appendix consider the estimation results of several TARCH(1,1)
specifications for SP500 weekly returns, to provide empirical context for the
rest of the paper. We obtained the daily levels of SP500 from CRSP database.
The total number of observations is 10,699 and covers the period from
07/02/1962 till 12/31/2004. We constructed weekly log returns and estimated
the parameters of TARCH and GARCH\ models with Gaussian and Student-t shocks
for 2 samples - full and post-1990.

Tables 1 and 2 shows estimated parameters and various data statistics. Note
that the Student-t distribution has an additional parameter, degrees of
freedom $\nu$, that adjusts the tails of the error distribution. Since the
Gaussian distribution is nested within the Student-t as a limit of large
degrees of freedom, and since the estimates of the full unconstrained model
result in a relatively small and statistically significant value of the
degrees of freedom, we conclude that the data points toward the fat-tailed
return shock distribution.

On the other hand, the asymmetric TARCH model is nested within the symmetric
GARCH in the limiting case $\alpha_{d}=0$. The estimated asymmetric
coefficient $\alpha_{d}$ in the TARCH\ model is not only non-zero, but
significantly higher than the symmetric coefficient $\alpha$ for both complete
and post 1990 samples, both daily and weekly frequencies and Gaussian and
Student-t shock distributions. Thus, we conclude that the asymmetric
volatility is prominently present in the data. The best fit model among those
considered is the TARCH(1,1) with Student-t distribution of return
innovations. The additional parameters of this model are statistically significant.

To make sure that asymmetry in volatility is not a result of several extreme
negative returns like 1987 crash we provide data statistics and re-estimated
parameters of TARCH\ models for trimmed full and post 1990 samples. The
trimming is done by cutting excess volatility in the most extreme 0.1\%
observations of both positive and negative return.

\begin{center}%
\begin{tabular}
[c]{cc|cccc|ccccc}%
\multicolumn{11}{c}{Table 1: SP500 moments.}\\\hline\hline
& Sample period & \multicolumn{4}{|c|}{Daily} & \multicolumn{4}{|c}{Weekly} &
\\\hline
&  & $s_{r}$ & $s_{r}^{d}$ & $k_{r}$ & $v_{r}^{d}$ & $s_{r}$ & $s_{r}^{d}$ &
$k_{r}$ & $v_{r}^{d}$ & \\
& 1962-2004 & {\small -1.40} & {\small -2.43} & {\small 39.83} & {\small 0.53}
& {\small -0.55} & {\small -1.35} & {\small 7.01} & {\small 0.55} & \\
& 1990-2004 & {\small -0.11} & {\small -1.14} & {\small 6.67} & {\small 0.51}
& {\small -0.64} & {\small -1.36} & {\small 6.10} & {\small 0.56} &
\\\hline\hline
\multicolumn{11}{c}{{\small SP500 moments(After trimming 0.1\% of extreme
positive and negative returns )}}\\\hline
& 1962-2004 & {\small 0.05} & {\small -1.03} & {\small 5.95} & {\small 0.50} &
{\small -0.39} & {\small -1.18} & {\small 5.26} & {\small 0.54} & \\
& 1990-2004 & {\small 0.04} & {\small -1.01} & {\small 5.56} & {\small 0.50} &
{\small -0.50} & {\small -1.24} & {\small 5.22} & {\small 0.55} &
\\\hline\hline
\end{tabular}

\bigskip

Table 2.
\end{center}

{\small Estimated parameters of GARCH(1,1)/TARCH(1,1) with Gaussian/Student-t
shocks on weekly SP500 returns. The total number of return observations is 756
for post-1990 sample and 2,139 for the full sample starting in 1962.. Number
below each parameter estimate in parenthesis is an asymptotic standard
deviation, LogL\ is corresponding loglikelihood value. GARCH\ parameters
correspond to the volatility specification: }$\sigma_{t}^{2}=\omega+\alpha
r_{t-1}^{2}+\alpha_{d}r_{t-1}^{2}1_{\left\{  r_{t-1}\leq0\right\}  }%
+\beta\sigma_{t-1}^{2}.$ ${\normalsize \nu}${\small is degrees of freedom of
Student-t distributed return innovation }$\varepsilon_{t}.$ {\small The
constant term of the volatility process is not shown since it is not used in
the simulations.}\bigskip

\begin{center}

\begin{tabular}
[c]{cccc|ccc}%
{\small Sample Dates} & \multicolumn{3}{c|}{{\small 01/01/1990-12/31/2004}} &
\multicolumn{3}{|c}{{\small 01/01/1962-12/31/2004}}\\\hline
{\small Model} & {\small GARCH} & {\small TARCH} & {\small TARCH + t} &
{\small GARCH} & {\small TARCH} & {\small TARCH + t}\\\hline
${\normalsize \alpha}$ & {\small 0.044} & {\small 0.007} & {\small 0.004} &
{\small 0.107} & {\small 0.037} & {\small 0.032}\\
& $\left(  0.0073\right)  $ & $\left(  0.024\right)  $ & $\left(
0.019\right)  $ & $\left(  0.013\right)  $ & $\left(  0.016\right)  $ &
$\left(  0.0125\right)  $\\
${\normalsize \alpha}_{d}$ & {\small -} & {\small 0.112} & {\small 0.094} &
{\small -} & {\small 0.136} & {\small 0.106}\\
&  & $\left(  0.046\right)  $ & $\left(  0.033\right)  $ &  & $\left(
0.033\right)  $ & $\left(  0.0223\right)  $\\
${\normalsize \beta}$ & {\small 0.953} & {\small 0.918} & {\small 0.927} &
{\small 0.886} & {\small 0.877} & {\small 0.894}\\
& $\left(  0.0003\right)  $ & $\left(  0.0022\right)  $ & $\left(
0.0013\right)  $ & $\left(  0.0022\right)  $ & $\left(  0.0031\right)  $ &
$\left(  0.0015\right)  $\\
${\normalsize \nu}$ & - & - & {\small 8.31} & - & - & {\small 10.19}\\
&  &  & $\left(  2.75\right)  $ &  &  & $\left(  2.42\right)  $\\
{\small LogL} & {\small 1855.6} & {\small 1861.7} & {\small 1877.4} &
{\small 5347.4} & {\small 5368.6} & {\small 5397.7}\\\hline\hline
\end{tabular}

\end{center}

\section{\label{appendix_LHP_MonteCarlo}Monte Carlo Simulations}

Most of the numerical estimates for credit risk in this paper are obtained by
Monte Carlo simulation. Here we outline the simulation procedure for two such
calculations, the estimation of the pairwise default correlation coefficient,
and the estimation of the tranche losses under the LHP assumption.

The default correlation coefficient, $\rho^{d}(p)$ for the factor GARCH and
TARCH models is calculated based on the simulated factor time series and
closed form formulas of conditional default probabilities:

\begin{itemize}
\item simulate the common factor, $R_{m,T}^{{}},$ $I=10,000$ times and
normalize it to have variance 1

\item for each $p$ find $d_{T}(p)$ that solves $\frac{1}{I}%
{\displaystyle\sum\limits_{i=1}^{I}}
\Phi\left(  \frac{d_{T}-bR_{m,T}^{(i)}}{\sqrt{1-b^{2}}}\right)  =p$

\item calculate $\rho^{d}(p)=\frac{p_{12}-p^{2}}{p(1-p)}$ where $p_{12}%
=\frac{1}{I}%
{\displaystyle\sum\limits_{i=1}^{I}}
\Phi\left(  \frac{d_{T}(p)-bR_{m,T}^{(i)}}{\sqrt{1-b^{2}}}\right)  ^{2}$
\end{itemize}

The likelihood bounds in Figure \ref{Figure_DefCorrBnds} were obtained by
repeating this procedure 1000 times.

When simulating portfolio loss distribution under the large homogeneous
portfolio (LHP)\ assumption we again begin by simulating the aggregated market
factor return. The latent variables are assumed to have symmetric one factor
structure with the factor following TARCH(1,1) model. For each realization of
the market factor the portfolio loss is given by the LHP formula
(\ref{LHP_L})
\begin{equation}
L_{T}=\left(  1-\overset{\_}{R}\right)  \Phi\left(  \frac{d_{T}-bR_{m,T}%
}{\sqrt{1-b^{2}}}\right) \nonumber
\end{equation}

where

\begin{itemize}
\item $R_{m,T}=%
{\displaystyle\sum\limits_{u=1}^{T}}
r_{m,u}/std\left(
{\displaystyle\sum\limits_{s=1}^{T}}
r_{m,s}\right)  $ is a normalised return over horizon T generated using time
aggregation of simulated TARCH(1,1) returns with unconditional volatility
equal to 1

\item $d_{T}$ is calibrated so that the probability of $R_{i,T}=bR_{m,T}%
+\sqrt{1-b^{2}}E_{T}$ hitting $d_{T}$ is equal to single name default
probability $p_{T}$%
\[
P\left(  bR_{m,T}+\sqrt{1-b^{2}}E_{T}\leq d_{T}\right)  =p_{T}%
\]

\item $b$ is the factor loading that is chosen to match a given unconditional
linear correlation $\rho=b^{2}$
\end{itemize}

To calculate the expected tranche losses generated by the model and to
calibrate $d_{T}$ we use $I=100,000$ independent Monte Carlo simulations of
the factor and then use corresponding sample moments:%
\begin{gather*}
d_{T}\text{ solves }\frac{1}{I}%
{\displaystyle\sum\limits_{i=1}^{I}}
\Phi\left(  \frac{d_{T}-bR_{m,T}^{(i)}}{\sqrt{1-b^{2}}}\right)  =p_{T}\\
EL_{\left(  0,K\right]  }=\frac{1}{I}%
{\displaystyle\sum\limits_{i=1}^{I}}
f_{\left(  0,K\right]  }\left(  \left(  1-\overset{\_}{R}\right)  \Phi\left(
\frac{d_{T}-bR_{m,T}^{(i)}}{\sqrt{1-b^{2}}}\right)  \right)
\end{gather*}


\begin{thebibliography}{99}                                                                                               %


\bibitem {Anderson-Sidenius}Anderson, L,. and J. Sidenius (2005): "Extensions
to Gaussian Copula: Random Recovery and Random Factor Loadings", \emph{Journal
of Credit Risk}, vol. 1, no. 1, p. 29

\bibitem {Black}F. Black, "Studies of stock price volatility changes",
\emph{Proceedings of the 1976 American Statistical Association, Business and
Economical Statistics Section}, p. 177

\bibitem {BollerslevGARCH}Bollerslev, T. (1986), "Generalized Autoregressive
Conditional Heteroskedasticity", \emph{Journal of Econometrics}, 31, 307-327.

\bibitem {Campb}Campbell, J. Y., and L. Hentschell, 1992, "No News is Good
News", \emph{Journal of Financial Economiscs}, 31, pp. 281-318

\bibitem {Diebold}Diebold, F., "Empirical Modeling of Exchange Rate Dynamics",
New York: Springer-Verlag, 1988.

\bibitem {Engle-Bollerslev-Nelson}Bollerslev, T., R.F. Engle, and D. Nelson
(1994): "ARCH Models", in \emph{The Handbook of Econometrics}, vol. IV, eds.
D.F. McFadden and R.F. Engle III. North Holland.

\bibitem {Day-Lewis-1992}Day, T. E., and C. M., Lewis (1992): "Stock Market
Volatility and the Informational Content of Stock Index Options."
\emph{Journal of Econometrics}, vol. 52, p. 267-287.

\bibitem {Derman-Kani-1994}Derman, E. and I. Kani (1994): "Riding on a Smile",
\emph{Risk}, vol. 7, no. 2, p. 32

\bibitem {Dupire-1994}Dupire, B. (1994): "Pricing with a Smile", \emph{Risk},
vol. 7, no. 1, p. 18

\bibitem {Drost-Nijman}Drost, F. C. and T. E. Nijman (1993): Temporal
aggregation of GARCH processes. \emph{Econometrica} vol. 61, p. 909--927.

\bibitem {Embrechts-Lindskog-McNeil-2001}Embrechts, P., F. Lindskog and A.
McNeil (2001): "Modeling dependence with copulas and applications to risk
management", working paper, ETH, Zurich

\bibitem {EngleArch}Engle, R.F. (1982): Autoregressive conditional
heteroscedasticity with estimates of the variance of United Kingdom in\ddag
ation. \emph{Econometrica}, vol. 50, p. 987-1006.

\bibitem {EngleSkew}Engle, R. F.(2004): \textquotedblleft Risk and Volatility:
Econometric Models and Financial Practice,\textquotedblright\ Nobel Lecture,
\emph{American Economic Review}, June 2004, V94, No.3

\bibitem {Frey-McNeil-2001}Frey, R., and A. McNeil, "Dependent defaults in
models of portfolio credit risk", \emph{Journal of Risk}\textit{ }(forthcoming),

\bibitem {GJR}Glosten, L.R., R. Jagannathan and D. Runkle (1993). On the
relation between the expected value and the volatility of the nominal excess
return on stocks. \emph{Journal of Finance}, 48, 1779-1801.

\bibitem {Gregory-Laurent-2003}Laurent, J.-P., \& J. Gregory (2003): "Basket
Default Swaps, CDOs and Factor Copulas", ISFA Actuarial School, University of
Lyon, working paper

\bibitem {Gregory(2004)}Gregory, J. (editor) (2003): "Credit Derivatives: The
Definitive Guide", Risk Books, Risk Waters Group, UK

\bibitem {GordyM}Gordy, M. (2000): \textquotedblleft A comparative anatomy of
credit risk models,\textquotedblright\ \emph{Journal of Banking and Finance},
vol. 24, p. 119--149.

\bibitem {HeTres}He, C., and T. Ter\"{a}svirta (1999): \textquotedblleft
Properties of moments of a family of GARCH processes,\textquotedblright%
\ \emph{Journal of Econometrics}, vol. 92, p. 173--192

\bibitem {HullWhite2004}Hull, J., and A. White (2004): "Valuation of a CDO and
an nth to Default CDS without Monte Carlo Simulation", \emph{Journal of
Derivatives}, vol. 12, no. 2, p. 8-23.

\bibitem {Lando}Lando, D., "Credit Risk Modeling -- Theory and Applications",
Princeton University Press. (2004)

\bibitem {Li-2000}Li, D. X. \ (2000): "On Default Correlations: a Copula
Function Approach", \emph{Journal of Fixed Income}, vol. 9, p. 43-54

\bibitem {Mashal-Naldi-Zeevi-2003}Mashal, R., Naldi, M. and Zeevi, A. (2003):
"Extreme Events and Multi-name Credit Derivatives",. in "Credit Derivatives
The Definitive Guide", Risk Waters Group, UK, p. 313-338

\bibitem {Merton1974}Merton, R (1974): "On pricing of corporate debt: the risk
structure of interest rates", \emph{Journal of Finance}, vol. 29, p. 449

\bibitem {JPMorganGuide}McGinty, L., E. Beinstein, R. Ahluwalia, and M. Watts
(2004): "Credit Correlation: A Guide", JPMorgan Credit Derivatives Strategy

\bibitem {Nelson}Nelson, D. B., 1991, "Conditional Heteroskedasticity in Asset
Returns: A New Approach",\textit{ }\emph{Econometrica}, 59, pp. 347-370

\bibitem {SchonbucherBook}Shonbucher, P. J. (2003): "Credit Derivatives
Pricing Models: Models, Pricing and Implementation", J. Wiley \& Sons, New York

\bibitem {Schonbucher-Shubert-2001}Schonbucher, P. J. and D. Schubert. (2001):
"Copula-dependent default risk in intensity models",. Working Paper,
Department of Statistics, Bonn University.

\bibitem {OKane-Schloegl-2003}O'Kane, D. and L. Schloegl (2003): "A Note on
the Large Homogeneous Portfolio Approximation with the Student-t
Copula",\textit{ \emph{Finance and Stochastics},} Volume 9, No 4, p 577 - 584.

\bibitem {Rubinstein-1994}Rubinstein, M. (1994), "Implied Binomial Trees",
\emph{Journal of Finance}, vol. 49, p. 771

\bibitem {Jackwerth-Rubinstein-1996}Jackwerth J.C. and M. Rubinstein (1996):
"Recovering Probability Distributions from Option Prices", \emph{Journal of
Finance}, vol. 51, p. 1611

\bibitem {Sklar}Sklar, A. (1959): "Fonctions de r epartition a n dimensions et
leurs marges", \emph{Publications de l`Institut de Statistique de l`Universit
e de Paris}, vol. 8, p. 229-231.

\bibitem {Vasicek-LHP}Vasicek, O. (1991), "Limiting Loan Loss Probability
Distribution", working paper, KMV Corp.

\bibitem {Vasicek-BiNormalExpansion}Vasicek, O. (1998): "A Series Expansion
for the Bivariate Normal Integral", working paper, KMV Corp.

\bibitem {Zakoian}Zako\"{\i}an, J.-M. (1994): "Threshold Heteroskedastic
Models", \emph{Journal of Economic Dynamics and Control}, vol. 18, p. 931-955.
\end{thebibliography}
\end{document}